\documentclass[11pt,a4paper,reqno]{article}
\usepackage[margin=3cm]{geometry}
\usepackage{authblk}
\usepackage[utf8]{inputenc}
\usepackage[T1]{fontenc}
\usepackage[nolist]{acronym}

\usepackage{amsmath, amssymb, amsfonts, amsthm}
\usepackage[all]{xy}
\usepackage{hyperref}
\usepackage{url}
\usepackage{mathtools}
\usepackage{calc}  
\usepackage{enumerate}
\usepackage{enumitem}
\usepackage{float}
\usepackage{bbm}
\usepackage{verbatim}
\usepackage{anyfontsize}
\usepackage{graphicx}
\graphicspath{{images/}{../images/}}
\usepackage{subcaption}
\usepackage{tabularx}
\usepackage{array}
\usepackage{wrapfig}
\usepackage{longtable}
\usepackage{makecell}
\usepackage{multirow, multicol}
\usepackage{ltablex}\keepXColumns

\newcolumntype{Y}{>{\centering\arraybackslash}X}
\newcolumntype{Z}{>{}Y}
\usepackage{rotating}
\usepackage{tikz}
\usetikzlibrary{shapes,arrows,positioning,arrows.meta,chains,positioning,quotes,fit}
\usepackage{blkarray}
\usepackage{tkz-euclide}
\usetikzlibrary{patterns} % LATEX and plain TEX when using Tik Z
\usepackage{mathrsfs}
\usetikzlibrary{fadings}
\usetikzlibrary{shadows.blur}
\usetikzlibrary{spy}
\usetikzlibrary{decorations.pathreplacing}

\usepackage{nicematrix}

\usepackage{pgfplots}
\usepgfplotslibrary{colorbrewer}
\pgfplotsset{compat=1.17}

\usepackage{blindtext}
\usepackage{chngcntr}
\usepackage[textwidth=\marginparwidth,textsize=tiny]{todonotes}
\usepackage{diagbox}
\mathtoolsset{showonlyrefs}
\usepackage{centernot}

\usepackage[final]{pdfpages}

\usepackage{appendix}

\numberwithin{equation}{section}

\numberwithin{equation}{section}

\theoremstyle{definition}
\newtheorem{theorem}{Theorem}[section]
\newtheorem{corollary}[theorem]{Corollary}
\newtheorem{proposition}[theorem]{Proposition}
\newtheorem{definition}[theorem]{Definition}
\newtheorem{example}[theorem]{Example}
\newtheorem{notation}[theorem]{Notation}
\newtheorem{remark}[theorem]{Remark}
\newtheorem{lemma}[theorem]{Lemma}

\makeatletter

\newcommand{\Rmnum}[1]{\expandafter\@slowromancap\romannumeral #1@}
\makeatother

\newcommand{\F}{\mathbb{F}}

\newcommand{\supp}{\operatorname{supp}}

\def\cov{\mathrel{\prec\kern-.6em\raise.015ex\hbox{$\;\cdot$}}}

\newcommand{\tuple}[1]{\mathbf{#1}}

\newcommand{\Ring}{\mathcal{R}}

\newcommand{\code}{\mathcal{C}}

\DeclareMathOperator{\wt}{wt}

\DeclareMathOperator{\rk}{rk}

\newcommand{\wtH}{\wt^{\mathsf{H}}}

\DeclareMathOperator{\dist}{d}

\newcommand{\HD}{\dist^{\mathsf{H}}}
\newcommand{\HomD}{\dist^{\mathsf{Hom}}}
\newcommand{\LD}{\dist^{\mathsf{L}}}

\newcommand{\suppH}{\supp^{\mathsf{H}}}

\newcommand{\LW}{\wt^{\mathsf{L}}}

\newcommand{\HomW}{\wt^{\mathsf{Hom}}}
\newcommand{\card}[1]{\left\vert \, {#1} \, \right\vert} % Cardinality
\newcommand{\st}{\, | \,}

\newcommand{\set}[1]{\left\lbrace{#1}\right\rbrace} % set

\newcommand{\ideal}[1]{\left\langle{#1}\right\rangle}

\DeclareMathOperator{\subtype}{\text{sbt}}

\usepackage{stmaryrd}

\newcommand{\gb}{\genfrac{[}{]}{0pt}{}}

\renewcommand{\and}{\text{ and }}

\colorlet{myteal}{teal!75!yellow}

\newcommand\latticeThreeD[1][1cm]{%
    \foreach \i in {0, 1, 2, 3}{
        \node at (0.2*\i, \i) (R\i) {\tiny$\bullet$};
        \node at (3.75+0.8*\i, 0.5*\i - 0.3) (L\i) {\tiny$\bullet$};
    }
    \foreach \i in {0,1,2}{
        \foreach \j in {0,1}{
            \node at (1.25 + 0.2*\i + 0.8*\j, \i - 0.1 + 0.5*\j) (RM\i\j) {\tiny$\bullet$};
        }
    }
    \foreach \j in {0,1,2}{
        \foreach \i in {0,1}{
            \node at (2.5 + 0.2*\i + 0.8*\j, \i - 0.2 + 0.5*\j) (LM\i\j) {\tiny$\bullet$};
        }
    }
    \foreach \i in {0, 1, 2}{
        \edef\indexplus{\fpeval{\i + 1}}
        \draw [postaction={decorate}] (RM\i0.center) -- (R\i.center);
        \draw [postaction={decorate}] (RM\i1.center) -- (RM\i0.center);
        \draw [postaction={decorate}] (R\i.center) -- (R\indexplus.center);
        \draw [postaction={decorate}] (L\i.center) -- (LM0\i.center);
        \draw [postaction={decorate}] (L\indexplus.center) -- (L\i.center);
    }
    \foreach \i in {0, 1}{
        \edef\indexplus{\fpeval{\i + 1}}
        \draw [postaction={decorate}] (RM\i0.center) -- (RM\indexplus0.center);
        \draw [postaction={decorate}] (RM\i1.center) -- (RM\indexplus1.center);
        \draw [postaction={decorate}] (LM\i0.center) -- (RM\i0.center);
        \draw [postaction={decorate}] (LM\i1.center) -- (RM\i1.center);
        \draw [postaction={decorate}] (LM\i1.center) -- (LM\i0.center);
        \draw [postaction={decorate}] (LM\i2.center) -- (LM\i1.center);
    }
    \foreach \i in {0}{
        \edef\indexplus{\fpeval{\i + 1}}
        \draw [postaction={decorate}] (LM\i0.center) -- (LM\indexplus0.center);
        \draw [postaction={decorate}] (LM\i1.center) -- (LM\indexplus1.center);
        \draw [postaction={decorate}] (LM\i2.center) -- (LM\indexplus2.center);
    }
}
\newcommand\latticelabels[1][1cm]{%
    \node[left, font = \tiny] at (R3) {$(0, 0, 0, 3)$};
    \node[left, font = \tiny] at (R2) {$(0, 0, 1, 2)$};
    \node[left, font = \tiny] at (R1) {$(0, 0, 2, 1)$};
    \node[left, font = \tiny] at (R0) {$(0, 0, 3, 0)$};
    \node[below = 0.2, left, font = \tiny] at (RM00) {$(0, 1, 2, 0)$};
    \node[below = 0.2, left, font = \tiny] at (RM10) {$(0, 1, 1, 1)$};
    \node[below = 0.2, left, font = \tiny] at (RM20) {$(0, 1, 0, 2)$};
    \node[yshift = -2mm, font = \tiny] at (LM00) {$(0, 2, 1, 0)$};
    \node[xshift = 5mm, yshift = -1mm, font = \tiny] at (LM10) {$(0, 2, 0, 1)$};
    \node[below = 0.1, right, font = \tiny] at (L0) {$(0, 3, 0, 0)$};

    \node[xshift = 3mm, yshift = -2mm, font = \tiny] at (RM01) {$(1, 0, 2, 0)$};
    \node[xshift = 5mm, yshift = -2mm, font = \tiny] at (RM11) {$(1, 0, 1, 1)$};
    \node[xshift = 5mm, yshift = -2mm, font = \tiny] at (RM21) {$(1, 0, 0, 2)$};
    \node[xshift = 5mm, yshift = -2mm, font = \tiny] at (LM11) {$(1, 1, 0, 1)$};
    \node[xshift = 5mm, yshift = -2mm, font = \tiny] at (LM01) {$(1, 1, 1, 0)$};
    \node[xshift = 5mm, yshift = -2mm, font = \tiny] at (L1) {$(1, 2, 0, 0)$};

    \node[xshift = 5mm, yshift = -2mm, font = \tiny] at (LM02) {$(2, 0, 1, 0)$};
    \node[xshift = 5mm, yshift = -2mm, font = \tiny] at (LM12) {$(2, 0, 0, 1)$};
    \node[xshift = 5mm, yshift = -2mm, font = \tiny] at (L2) {$(2, 1, 0, 0)$};

    \node[xshift = 5mm, yshift = -2mm, font = \tiny] at (L3) {$(3, 0, 0, 0)$};
}
\newcommand\latticeCUBE[1][1cm]{%
    \foreach \i in  {0}{
    \edef\indexiplus{\fpeval{\i + 1}}
        \foreach \j in {0,1}{
            \draw[glow] (RM\i\j.center) -- (RM\indexiplus\j.center);
            \draw[glow] (LM\i\j.center) -- (LM\indexiplus\j.center);
            \draw[glow] (RM\i\j.center) -- (RM\i\indexiplus.center);
            \draw[glow] (LM\i\j.center) -- (LM\i\indexiplus.center);
            \draw[glow] (RM\indexiplus\j.center) -- (RM\indexiplus\indexiplus.center);
            \draw[glow] (LM\indexiplus\j.center) -- (LM\indexiplus\indexiplus.center);
            \draw[glow] (RM\i\j.center) -- (LM\i\j.center);
            \draw[glow] (RM\indexiplus\j.center) -- (LM\indexiplus\j.center);
        }
    }
}
\newcommand\latticeSQUARE[1][1cm]{%
    \draw[glow] (RM00.center) -- (RM10.center) -- (LM10.center) -- (LM00.center) -- (RM00.center);
}
\newcommand\latticeLINE[1][1cm]{%
    \draw[glow] (RM10.center) -- (RM11.center);
}

\author[1]{Jessica Bariffi \thanks{J. B. is funded by the European Union, grant n. DiDAX, 101115134.}}
\affil[1]{Technical University of Munich, Munich, Germany}

\author[2]{Drisana Bhatia\thanks{D. B. was partially supported by a UC Berkeley SURF Fellowship.}}
\affil[2]{University of California Berkeley, CA, U.S.A.}

\author[3]{Giuseppe Cotardo}
\affil[3]{Virginia Tech, Blacksburg, VA, U.S.A.}

\author[1]{Violetta Weger}

\title{Weak Composition Lattices and Ring-Linear Anticodes}

\date{}

\begin{document}
\maketitle

\begin{abstract}
Lattices and partially ordered sets have played an increasingly important role in coding theory, providing combinatorial frameworks for studying structural and algebraic properties of error-correcting codes. Motivated by recent works connecting lattice theory, anticodes, and coding-theoretic invariants, we investigate ring-linear codes endowed with the Lee metric. We introduce and characterize optimal Lee-metric anticodes over the ring $\mathbb{Z}/p^s\mathbb{Z}$. We show that the family of such anticodes admits a natural partition into subtypes and forms a lattice under inclusion. We establish a bijection between this lattice and a lattice of weak compositions ordered by dominance. As an application, we use this correspondence to introduce new invariants for Lee-metric codes via an anticode approach.
\end{abstract}
\maketitle

\section{Introduction}

Lattices and partially ordered sets arise naturally in many areas of discrete mathematics, including optimization and enumerative combinatorics. A general reference is~\cite{stanley2011enumerative}. Over the past decades, deep connections between lattice theory and coding theory have emerged, revealing that combinatorial structures can be used to study fundamental properties of error-correcting codes and, conversely, that techniques from coding theory can inform the study of lattices.
A central contribution in this direction is due to Dowling, who introduced the higher-weight Dowling lattices in connection with one of the central problems of classical coding theory: determining the largest possible dimension of a code in the Hamming metric with prescribed length and error-correction capability \cite{dowling1971codes,dowling1973q}. These lattices were further studied by other authors, for example in~\cite{bonin1993modular, brini1982some, kung1996critical, ravagnani2019whitney}, particularly in connection with Segre’s famous conjecture on the largest size of an arc in a projective space over a finite field~\cite{segre1955curve}. A $q$-analogue of Dowling's theory was established in \cite{cotardo2023rank,cotardo2025whitney} in the rank-metric setting, where analogous extremal problems arise. In that context, the associated lattices provide a natural framework for defining and studying invariants that encode both combinatorial and algebraic properties of rank-metric codes.

In~\cite{byrne2023tensor}, it was shown that families of \textit{anticodes} naturally form a lattice structure, and that exploiting this structure leads to new invariants for tensor rank-metric codes, as well as a general formulation of the MacWilliams identities. A similar approach was adopted in~\cite{cotardo2023rank} to derive new invariants for vector rank-metric codes. Families of anticodes were previously characterized in~\cite{ravagnani2016generalized} for Hamming-metric and rank-metric codes, where they were used to derive invariants for codes.
More recently, the anticode framework has been extended to the quantum setting. In~\cite{cao2025quantum}, the authors introduced anticodes for quantum error-correcting codes modeled as symplectic spaces, showing that this framework captures structural properties of the codes and provides an algebraic interpretation of principles originating in quantum physics.

In this paper, we focus on codes over rings, and in particular on codes endowed with the Lee metric. In this context, a code is a $\mathbb{Z}/p^s\mathbb{Z}$-submodule. The Lee metric was originally introduced by Lee~\cite{lee1958some} to model phase modulation and can be viewed as a natural generalization of the Hamming metric beyond the binary field. It gained prominence following the seminal result in \cite{sole} showing that certain optimal but non-linear binary codes can be realized as Gray images of linear codes over $\mathbb{Z}/4\mathbb{Z}$ equipped with the Lee metric. More recently, the Lee metric has attracted renewed interest due to applications in cryptography \cite{restballs, al, fuleeca,leenp}, as well as its rich combinatorial structure \cite{golomb}. Several aspects of Lee-metric codes have been studied, including constant-Lee-weight codes and bounds showing the non-existence of nontrivial MDS-type codes \cite{byrne}, for instance with respect to Shiromoto’s Singleton bound~\cite{shiromoto2000singleton} or the bounds provided in \cite{newleesb}. Despite this extensive literature, a systematic theory of Lee-metric anticodes has so far been missing.   

We introduce and study anticodes for ring-linear codes endowed with the Lee metric. In our framework, an anticode is defined as a ring-linear code whose rank equals its maximum weight, in analogy with~\cite{ravagnani2016generalized}. We provide a characterization of the anticodes that are optimal with respect to the Lee metric. A key structural result of this work is that the family of optimal Lee-metric anticodes admits a natural partition with respect to their \textit{subtypes}, and that these families form a lattice under inclusion. We prove that this lattice is in bijection with a lattice of weak compositions ordered by dominance. This combinatorial approach is inspired by Brylawski’s work~\cite{brylawski1973lattice} on the lattice of integer partitions ordered by dominance. We generalizes Brylawski’s lattice in two directions: we allow nonnegative components and we fix the number of components. We show that the lattice of weak compositions under dominance order is connected with the lattice of optimal Lee-metric anticodes over the ring $\mathbb{Z}/p^s\mathbb{Z}$ ordered by inclusion. We remark that in~\cite{panek2021optimal}, the authors introduce a notion of anticodes for distance-regular graphs and study this framework in connection with other metrics, including the Lee metric. However, to the best of our knowledge, the notion of anticodes considered there and the results presented in this paper are different, as they rely on a different notion of optimality.

The paper is organized as follows. In Section~\ref{sec:prelim}, we introduce notation and recall basic notions on posets and lattices. Section~\ref{sec:new-lattices} introduces the lattice of weak compositions ordered by dominance, studies its structural properties, and derives its Möbius function. In Section~\ref{sec:anticode}, we first recall relevant definitions and results on ring-linear codes, and then introduce and characterize optimal Lee-metric anticodes over $\mathbb{Z}/p^s\mathbb{Z}$. We also establish the correspondence between the lattice of optimal Lee-metric anticodes (ordered by inclusion) and the lattice of weak compositions (ordered by dominance). Finally, in Section~\ref{sec:invariants}, we use this correspondence to define new invariants for Lee-metric codes via the anticode approach. As an application, we derive generalized ring weights and show how they relate to other well-studied invariants.

\section{Preliminaries and Notation}\label{sec:prelim}

In this section, we briefly review some fundamental definitions and results that will be used throughout the paper.  A standard reference is \cite[Chapter~3]{stanley2011enumerative}. We begin with some background on posets and lattices. A \textbf{partially ordered set} (or \textbf{poset}) is a pair $(\mathscr{P}, \leq)$ where $\mathscr{P}$ is a nonempty set and~$\leq$ is a binary relation satisfying the three axioms of \textit{reflexivity}, \textit{antisymmetry}, and \textit{transitivity}. In the following, we abuse the notation, and write $\mathscr{P}$ for $(\mathscr{P}, \leq)$. Moreover, for  $x,y\in\mathscr{P}$, we write $x<y$ for $x\leq y$ and $x\neq y$.

\begin{definition}
   Let $x, y \in \mathscr{P}$, we say that $x$ and $y$ are \textbf{comparable} if $x \leq y$ or $y \leq x$.  We say that $y$ \textbf{covers} $x$ (or that $x$ \textbf{is covered by} $y$) if $x< y$ and there is no $z\in \mathscr{P}$ such that $x<z<y$. 
\end{definition}

\begin{definition}
The \textbf{join} of $x,y\in\mathscr{P}$, should it exists, is the element $z\in\mathscr{P}$ that satisfies $x\leq z$, $y\leq z$ and $z\leq t$ for any $t\in\mathscr{P}$ with $x\leq t$ and $y\leq t$. Dually, the \textbf{meet} of $x,y\in\mathscr{P}$, should it exists, is the element $z\in\mathscr{P}$ that satisfies $z\leq x$, $z\leq y$ and $t\leq z$ for any $t\in\mathscr{P}$ with $t\leq x$ and $t\leq y$. The poset $\mathscr{P}$ has \textbf{minimum element} $0_{\mathscr{P}} \in\mathscr{P}$ if $0_{\mathscr{P}}\leq x$ for any $x\in\mathscr{P}$. Dually, it has \textbf{maximum element} $1_{\mathscr{P}} \in\mathscr{P}$ if $x\leq 1_{\mathscr{P}}$  for any $x\in\mathscr{P}$.
\end{definition}

It is easy to verify that the join (respectively, the meet) of $x, y \in \mathscr{P}$, should it exist, is unique and we denote it by $x \vee y$ (respectively, $x \wedge y$).  
The minimum (respectively, the maximum) element of $\mathscr{P}$, should it exist, is unique and if the poset is clear from the context we simply denote it by $0$ (respectively, $1$).

\begin{definition}
The \textit{Möbius function} of $\mathscr{P}$ is defined recursively as
\begin{equation*}
\mu_{\mathscr{P}}(x,y) = 
\begin{cases}
1 & \text{if } x = y, \\
\displaystyle -\sum_{x \leq z < y} \mu_{\mathscr{P}}(x,z) & \text{if } x < y, \\
0 & \text{otherwise}.
\end{cases}
\end{equation*} 
We  write $\mu_\mathscr{P}(x)$ for $\mu_{\mathscr{P}}(0,x)$, when $\mathscr{P}$ has minimum element $0$.
\end{definition}

The following result is central in combinatorics.

\begin{proposition}[Möbius Inversion Formula]
    Let $\mathbb{K}$ be a field and let $f:\mathscr{P}\rightarrow\mathbb{K}$ be a function. Define $g:\mathscr{P}\rightarrow\mathbb{K}$
	by $g(y)=\sum_{x\leq y}f(x)$ for all $y \in \mathscr{P}$.
	Then
	\begin{equation*}
	f(y)=\sum_{x\leq y}\mu_\mathscr{P}(x,y)\, g(x) \quad \mbox{for all $y \in \mathscr{P}$}.
	\end{equation*}
\end{proposition}

\begin{definition}
A \textbf{chain} in $\mathscr{P}$ is a non-empty subset $C \subseteq \mathscr{P}$ such that every $s,t \in C$ are comparable. We say that $C$ is \textbf{maximal} if there is no chain in $\mathscr{P}$ that strictly contains~$C$. We say that $\mathscr{P}$ is \textbf{graded} if all maximal chains in $\mathscr{P}$ have the same cardinality.
\end{definition}

This work focuses on a particular class of posets known as \textit{lattices}. Many well-known posets, including those formed by the subsets of a finite set and by the linear subspaces of a vector space over a finite field, are in fact lattices.

\begin{definition}
A \textbf{lattice} $\mathscr{L}$ is a poset that contains the join and the meet of every pair of its elements. In this case, join and meet can be seen as commutative and associative operations $\vee, \wedge : \mathscr{L} \times \mathscr{L} \to \mathscr{L}$. In particular, the \textbf{join} (respectively, the \textbf{meet}) of a non-empty subset $S \subseteq \mathscr{L}$ is well-defined as the join (respectively, the meet) of its elements and denoted by~$\vee S$ (resp., by~$\wedge S$). 
Every finite lattice has a minimum and a maximum element, denoted by $0_{\mathscr{L}}$ and $1_{\mathscr{L}}$ respectively. The elements $x\in\mathscr{L}$ that cover $0_\mathscr{L}$ are called \textbf{atoms}.
\end{definition}

A lattice $\mathscr{L}$ is \textbf{atomistic} if every element of $\mathscr{L}$ is the join of a set of atoms of $\mathscr{L}$.

\begin{definition}
    A subset $\mathscr{M}\subseteq \mathscr{L}$ is a \textbf{sublattice} if it is closed under the operations of $\vee$ and $\wedge$ in~$\mathscr{L}$.   
\end{definition}

From a combinatorial standpoint, one of the most important families of lattices is the class of \textit{distributive lattices}. Recall that a lattice $\mathscr{L}$ is \textbf{distributive} if the following distributive laws hold for all~$x, y, z \in \mathscr{L}$:
\begin{enumerate}
    \item $x \vee (y \wedge z) = (x \vee y) \wedge (x \vee z)$;
    \item $x \wedge (y \vee z) = (x \wedge y) \vee (x \wedge z)$.
\end{enumerate}
It is not difficult to verify that either of these laws implies the other. 

\begin{definition}
    Let $x,y\in\mathscr{L}$. We say that $y$ is a \textbf{complement} of $x$ if $x\wedge y=0$ and $x\vee y=1$.
\end{definition}

It is well known that, in a distributive lattice $\mathscr{L}$, complements are unique. We denote by $x'$ the complement of an element $x \in \mathscr{L}$. 

\begin{definition}
    A \textbf{Boolean lattice} $\mathscr{B}$ is a lattice that contains the complement of every element. In this case, the complement can be seen as an operation $':\mathscr{B}\rightarrow\mathscr{B}$.
\end{definition}

Note that every Boolean lattice is atomistic. Examples include the power set lattice ordered by inclusion, the lattice of all binary strings of fixed length with component-wise order, and the lattice of subspaces of a vector space over $\mathbb{F}_2$ ordered by inclusion. We recall the definition of a partial order on sequences of nonnegative integers, known as the \textit{dominance order} (also referred to as the \textit{majorization order}).

\begin{definition}
    Let $\tuple{a}=(a_0,\ldots,a_n)$ and $\tuple{b}=(b_0,\ldots,b_n)$ be sequences of nonnegative integers. We say that $\tuple{a}$ \textbf{precedes} $\tuple{b}$ \textbf{in the dominance order}, we write $\tuple{a} \preceq \tuple{b}$, if for all $0\leq j\leq n$ we have
    \begin{equation*}
        \sum_{i=0}^ja_i\leq \sum_{i=0}^j b_i.
    \end{equation*}
\end{definition}

In the literature, the dominance order is often studied as an ordering on the lattice of \textit{partitions} (see, for example,~\cite{brylawski1973lattice,stanley2011enumerative}). In this work, we focus on the study of posets of compositions. We recall their definitions.

\begin{definition}
     An \textbf{$s$-composition} of a positive integer $n$ is a sequence $\tuple{a}=(a_0,\ldots,a_{s-1})$, with $s \geq 1$, of positive integers satisfying $\sum_{i=0}^{s-1}a_i=n$.
\end{definition}

\begin{definition}
     A \textbf{weak $s$-composition} of a nonnegative integer $n$ is a sequence $\tuple{a}=(a_0,\ldots,a_{s-1})$, with $s \geq 1$, of nonnegative integers satisfying $\sum_{i=0}^{s-1}a_i=n$. We denote by $\Delta_s(n)$ the set of weak compositions of length~$s$.
\end{definition}

\begin{definition}
     An \textbf{$s$-partition} of a positive integer $n$ is a composition $\tuple{a}=(a_0,\ldots,a_{s-1})$ of $n$ that is non-increasing, i.e., $a_0\geq a_1\geq \cdots\geq a_{s-1}$.
\end{definition}
For all three, compositions, weak compositions and partitions, we use the convention that $n=0$ only has the empty partition, denoted by $\emptyset$. 
For example, the set $\Delta_3(4)$ of $3$-weak composition of $4$ has the following elements
\begin{equation*}
\renewcommand{\arraystretch}{2}
   \begin{array}{ccccc}
        (4,0,0), &(0,4,0), &(0,0,4), &(1,3,0), &(1,0,3),\\
        (0,1,3), &(2,2,0), &(2,0,2), &(0,2,2), &(3,1,0),\\
        (3,0,1), &(0,3,1), &(1,1,2), &(1,2,1), &(2,1,1).
   \end{array}
\end{equation*}

\section{The Lattice of Weak Compositions}\label{sec:new-lattices}

In~\cite{brylawski1973lattice}, Brylawski introduced the lattice of $s$-partitions with respect to the dominance order.  In this section, we extend this theory and study posets of weak compositions of length $s$ with respect to the {dominance order}. We show that these posets are indeed lattices, determine their M\"obius function, and provide necessary and sufficient conditions for one element to cover another. Throughout this section, let $n$ and $s$ be nonnegative integers. We adopt the convention that $\sum_{i=0}^{-1} f(i) = 0$ for any function $f$. We begin by recalling the following definition from~\cite{brylawski1973lattice}.

\begin{definition}
    Let $\tuple{a}=(a_0,\ldots,a_{s-1})$ be a sequence of integers. The \textbf{ associated sequence} of $\tuple{a}$ is defined as~$\widehat{\tuple{a}}=(\widehat{a}_0,\ldots,\widehat{a}_{s-1})$ with $\widehat{a}_j:=\sum_{i=0}^ja_i$. We denote by $\widehat{\Delta}_s(n)$ the set of $s$-sequences associated with the elements of $\Delta_s(n)$.
\end{definition}

Note that, in~\cite{brylawski1973lattice}, the associated sequence of $\tuple{a}$ is an $(s+1)$-sequence $\alpha = (0 \mid \widehat{a})$, where $\mid$ denotes sequence concatenation.

\begin{remark}\label{rem:brylawski}
    As already observed in~\cite{brylawski1973lattice}, the $s$-sequence $\widehat{\tuple{a}}$ associated with an integer sequence $\tuple{a}$ of length $s$ is uniquely determined and satisfies the following properties: $\widehat{\tuple{a}}$ is non-decreasing if and only if $\tuple{a}$ is nonnegative, $\widehat{\tuple{a}}$ is concave, i.e., $2\widehat{a}_j\geq \widehat{a}_{j+1}+\widehat{a}_{j-1}$, if and only if $\tuple{a}$ is non-increasing, $\widehat{a}_{s-1} = n$ if and only if $\sum_{j=0}^{s-1} a_j = n$. Finally, $\tuple{a}$ can be recovered from $\widehat{\tuple{a}}$ using the relations $a_0 = \widehat{a}_0$ and $a_j = \widehat{a}_j - \widehat{a}_{j-1}$ for all $j \in \{1, \ldots, s-1\}$.
\end{remark}

\begin{remark}\label{rem:brylawski2}
    One can easily verify that the dominance order on $\Delta_s(n)$ corresponds to the component-wise order on the associated sequences in $\widehat{\Delta}_s(n)$.  In particular, for any $\tuple{a}, \tuple{b} \in \Delta_s(n)$, we have $\tuple{a} \preceq \tuple{b}$ if and only if~$\widehat{a}_j \leq \widehat{b}_j$ for all $j \in \{0, \ldots, s-1\}$. Moreover, the poset $\widehat{\Delta}_s(n)$, equipped with the component-wise order, forms a \textit{meet-semilattice}. We abuse notation and write $\widehat{\Delta}_s(n)$ to denote this meet-semilattice.
\end{remark}

The following result follows from the remarks above. We include a proof for completeness.

\begin{lemma}\label{lem:vee-wedge}
    Let $\tuple{a}, \tuple{b}\in\Delta_s(n)$. Their join and meet, respectively, are given by
    \begin{enumerate}
        \item\label{item1:vee-wedge} $\tuple{a} \vee \tuple{b} = \tuple{c}$ with $\displaystyle c_j=\max\left(\widehat{a}_j,\widehat{b}_j\right)-\widehat{c}_{j-1}$,
        \item\label{item2:vee-wedge} $\tuple{a} \wedge \tuple{b} = \tuple{d}$ with $\displaystyle d_j=\min\left(\widehat{a}_j,\widehat{b}_j\right)-\widehat{d}_{j-1}$.
    \end{enumerate}
\end{lemma}
\begin{proof}
    It is not hard to check that $\widehat{\tuple{a}} \vee \widehat{\tuple{b}} = (\max(\widehat{a}_0, \widehat{b}_0), \ldots, \max(\widehat{a}_{s-1}, \widehat{b}_{s-1}))$. By Remark~\ref{rem:brylawski}, we obtain $(\tuple{a} \vee \tuple{b})_0 = \max(\widehat{a}_0, \widehat{b}_0)$ and, for any $j \in \{1, \ldots, s-1\}$,
    \begin{equation*}
    (\tuple{a} \vee \tuple{b})_j = (\widehat{\tuple{a}} \vee \widehat{\tuple{b}})_j - (\widehat{\tuple{a}} \vee \widehat{\tuple{b}})_{j-1} = \max(\widehat{a}_j, \widehat{b}_j) - (\widehat{\tuple{a}} \vee \widehat{\tuple{b}})_{j-1} = \max\left( \widehat{a}_j, \widehat{b}_j \right) - (\widehat{\tuple{a}} \vee \widehat{\tuple{b}})_{j-1},
    \end{equation*}
which implies~\ref{item1:vee-wedge}. A similar argument establishes~\ref{item2:vee-wedge}.
\end{proof}

The following is one of the main result of this section.

\begin{theorem}
    $(\Delta_s(n), \preceq)$ is a distributive lattice with $1=(n,0,\ldots,0)$ and $0=(0,\ldots,0,n)$.
\end{theorem}
\begin{proof}
    Remarks~\ref{rem:brylawski} and~\ref{rem:brylawski2}, together with Lemma~\ref{lem:vee-wedge}, show that $(\Delta_s(n), \preceq)$ is closed under the operations $\vee$ and $\wedge$, and hence forms a lattice with $1=(n,0,\ldots,0)$ and $0=(0,\ldots,0,n)$. It remains to show that $(\Delta_s(n), \preceq)$ is distributive. Let $\widehat{\tuple{a}}, \widehat{\tuple{b}}, \widehat{\tuple{c}} \in \widehat{\Delta}_s(n)$ be the $s$-sequences associated with $\tuple{a}$, $\tuple{b}$, and $\tuple{c}$, for some $\tuple{a}, \tuple{b}, \tuple{c} \in \Delta_s(n)$. By the distributive properties of the $\min$ and $\max$ functions, we have
    \begin{align*}
        \widehat{\tuple{a}} \vee (\widehat{\tuple{b}} \wedge \widehat{\tuple{c}}) &= \left( \max(\widehat{a}_0, \min(\widehat{b}_0, \widehat{c}_0)), \ldots, \max(\widehat{a}_{s-1}, \min(\widehat{b}_{s-1}, \widehat{c}_{s-1})) \right) \\
        &= \left( \min(\max(\widehat{a}_0, \widehat{b}_0), \max(\widehat{a}_0, \widehat{c}_0)), \ldots, \min(\max(\widehat{a}_{s-1}, \widehat{b}_{s-1}), \max(\widehat{a}_{s-1}, \widehat{c}_{s-1})) \right) \\
        &= (\widehat{\tuple{a}} \vee \widehat{\tuple{b}}) \wedge (\widehat{\tuple{a}} \vee \widehat{\tuple{c}}).
    \end{align*}
    The result now follows from Remark~\ref{rem:brylawski2}. 
\end{proof}

In the remainder, we abuse notation and write $\Delta_s(n)$ to denote the lattice $(\Delta_s(n), \preceq)$. It is easy to check that the lattice of partitions introduced in~\cite{brylawski1973lattice} forms a sublattice of $\Delta_s(n)$.

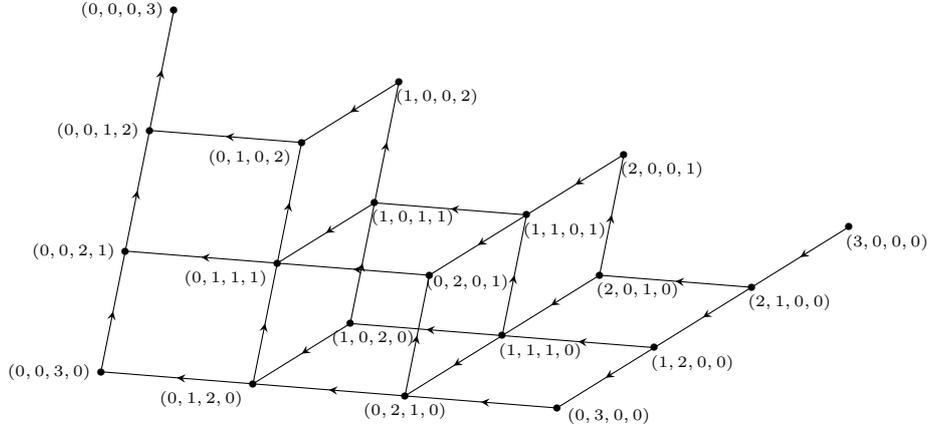
\begin{figure}[H]
    \centering
        \begin{tikzpicture}[scale = 1.6, decoration={markings, mark= at position 0.5 with {\arrow{stealth}}}]
            \latticeThreeD
            \latticelabels
        \end{tikzpicture}
        \caption{Graphical representation of the lattice $\Delta_4(3)$.}
        \label{fig:lattice3D_basic}
\end{figure}
\iffalse
\begin{figure}
\centering
    \begin{subfigure}{0.45\textwidth}
        \begin{tikzpicture}[decoration={markings, mark= at position 0.5 with {\arrow{stealth}}}, , glow/.style={%
        preaction={draw,line cap=round,line join=round,
        opacity=0.3,line width=4pt,#1}},glow/.default=cyan,
        transparency group]
            \latticeThreeD
            \latticeLINE
        \end{tikzpicture}
        \caption{Highlighted line}
        \label{fig:lattice3D_line}
    \end{subfigure}
    \begin{subfigure}{0.45\textwidth}
        \begin{tikzpicture}[decoration={markings, mark= at position 0.5 with {\arrow{stealth}}}, , glow/.style={%
        preaction={draw,line cap=round,line join=round,
        opacity=0.3,line width=4pt,#1}},glow/.default=cyan,
        transparency group]
            \latticeThreeD
            \latticeSQUARE
        \end{tikzpicture}
        \caption{Highlighted square}
        \label{fig:lattice3D_square}
    \end{subfigure}
    \begin{subfigure}{0.45\textwidth}
        \begin{tikzpicture}[decoration={markings, mark= at position 0.5 with {\arrow{stealth}}}, , glow/.style={%
        preaction={draw,line cap=round,line join=round,
        opacity=0.3,line width=4pt,#1}},glow/.default=cyan,
        transparency group]
            \latticeThreeD
            \latticeCUBE
        \end{tikzpicture}
        \caption{Highlighted cube}
        \label{fig:lattice3D_cube}
    \end{subfigure}
    \caption{Lattice drawings}
    \label{fig:lattices}
\end{figure}
\fi

\begin{notation}
    For any $\tuple{a}, \tuple{b} \in \Delta_s(n)$, we write $\tuple{a} \cov \tuple{b}$ to mean that $\tuple{b}$ covers $\tuple{a}$.
\end{notation}

The following result extends~\cite[Proposition~2.3]{brylawski1973lattice} to the setting of $\Delta_s(n)$. We note that one of the conditions in~\cite[Proposition~2.3]{brylawski1973lattice} must be relaxed when working with compositions.

\begin{theorem}\label{thm:cover}
   Let $\tuple{a}, \tuple{b} \in \Delta_s(n)$. Then $\tuple{a} \cov \tuple{b}$ if and only if there exists $j \in \{0, \ldots, s-2\}$ such that $b_j = a_j + 1$, $b_{j+1} = a_{j+1} - 1$, and $b_i = a_i$ for all $i \in \{0, \ldots, s-1\} \setminus \{j, j+1\}$.   
\end{theorem}
\begin{proof}
    Suppose that $\tuple{a}, \tuple{b} \in \Delta_s(n)$ satisfy the conditions in the statement. This translates to $\widehat{b}_j = \widehat{a}_j + 1$ and $\widehat{b}_i = \widehat{a}_i$ for all $i \in \{0, \ldots, s-1\} \setminus \{j\}$. It follows that $\widehat{b}$ covers $\widehat{\tuple{a}}$ in $\widehat{\Delta}_s(n)$, and therefore $\tuple{a} \cov \tuple{b}$ in~$\Delta_s(n)$. Conversely, suppose $\tuple{a} \cov \tuple{b}$. Then there exists $j \in \{0, \ldots, s-2\}$ such that $b_i = a_i$ for all $i \in \{0, \ldots, j-1\}$ and $b_j = a_j + 1$. Since $\tuple{a}$ and $\tuple{b}$ are both elements of $\Delta_s(n)$, their components must sum to $n$. Therefore, there must exist $k \in \{j+1, \ldots, s-1\}$ such that $b_k = a_k - 1$ and $b_i = a_i$ for all $i \in \{j+1, \ldots, s-1\} \setminus \{k\}$. It remains to show that $k = j+1$. Suppose, toward a contradiction, that $k > j+1$. Let $k' \in \{j+1, \ldots, k-1\}$ and define 
    \begin{equation*}
        \tuple{c} = (a_0, \ldots, a_j + 1, a_{j+1}, \ldots, a_{k'-1}, a_{k'} - 1, a_{k'+1}, \ldots, a_{s-1}).
    \end{equation*}
    One can check that $\tuple{c} \in \Delta_s(n)$ and that $\tuple{a} \prec \tuple{c} \prec \tuple{b}$. This contradicts the assumption that $\tuple{a} \cov \tuple{b}$. Therefore, we must have $k = j+1$, which completes the proof.
\end{proof}

For a sequence $\tuple{a}=(a_0, \ldots, a_{s-1})$, we denote by $\suppH(\tuple{a})$ and $\wtH(\tuple{a})$ the \textbf{Hamming support} and the \textbf{Hamming weight} of $\tuple{a}$, respectively, that is, $\suppH(\tuple{a}) = \{j \;:\; a_j \neq 0\}$ and $\wtH(\tuple{a}) = |\suppH(\tuple{a})|$. For any $j \in \{0, \ldots, s-1\}$, we denote by $e^{(j)}(s)$ the sequence of length $s$ with a $1$ in position $j$ and $0$ elsewhere. When the length is clear from context, we simply write $e^{(j)}$ instead of $e^{(j)}(s)$. Throughout the paper, for sequences $\tuple{a}=(a_0, \ldots, a_{s-1})$ and $\tuple{b}=(b_0, \ldots, b_{s-1})$, we use the following notation: $-\tuple{a}$ denotes the sequence $(-a_0, \ldots, -a_{s-1})$, $\tuple{a}+\tuple{b}$ denotes the component-wise sum $(a_0+b_0, \ldots, a_{s-1}+b_{s-1})$ and $\tuple{a}-\tuple{b}$ denotes the sequence $\tuple{a}+(-\tuple{b})$.

\begin{corollary}
    Let $\tuple{a}, \tuple{b} \in \Delta_s(n)$. Then $\tuple{a} \cov \tuple{b}$ if and only if $\tuple{b} - \tuple{a} = e^{(j)}-e^{(j+1)}$. Moreover, the number of elements in $\Delta_s(n)$ that cover $\tuple{a}$ is $\wtH((a_1,\ldots,a_{s-1}))$.
\end{corollary}
\begin{proof}
    The first part of the statement follows directly from Theorem~\ref{thm:cover}. For the second part, let $j \in \{0, \ldots, s-1\}$ and observe that $\tuple{a} + e^{(j)} - e^{(j+1)} \in \Delta_s(n)$ if and only if $a_{j+1} \neq 0$. This implies the statement.
\end{proof}

The following result shows that the lattice $\Delta_s(n)$ contains Boolean sublattices. 

\begin{proposition}\label{prop: bool}
    Let $a\in\Delta_s(n)$ and let 
    \begin{equation*}
        \mathscr{B}(\tuple{a})=\left\{\tuple{a}+\sum_{j\in S}\left(e^{(j-1)}-e^{(j)}\right) \;:\; S \subseteq \suppH((a_1,\ldots,a_{s-1}))\right\}.
    \end{equation*}
    We have that $(\mathscr{B}(\tuple{a}),\preceq)$ is a Boolean sublattice of $\Delta_s(n)$. 
\end{proposition}
\begin{proof}
    We first observe that each element of $\mathscr{B}(\tuple{a})$ is uniquely determined by a subset of~$\suppH((a_1,\ldots,a_{s-1}))$. It is not hard to check that $(\mathscr{B}(\tuple{a}), \preceq)$ is a sublattice of $\Delta_s(n)$. It remains to show that $(\mathscr{B}(\tuple{a}), \preceq)$ is a Boolean lattice. Let $\tuple{b}, \tuple{c} \in \mathscr{B}(\tuple{a})$ and define $B, C \subseteq \suppH((a_1,\ldots,a_{s-1}))$ such that
    \begin{equation*}
        \tuple{b} = \tuple{a} + \sum_{j \in B} \left(e^{(j-1)} - e^{(j)}\right) \quad \text{and} \quad \tuple{c} = \tuple{a} + \sum_{j \in C} \left(e^{(j-1)} - e^{(j)}\right).
        \end{equation*}
   It is not hard to check that $b \preceq c$ in $\mathscr{B}(\tuple{a})$ if and only if $B \subseteq C$. Therefore, $(\mathscr{B}(\tuple{a}), \preceq)$ is isomorphic to the power set lattice of $\suppH((a_1,\ldots,a_{s-1}))$ ordered by inclusion, and is thus a Boolean lattice.
\end{proof}

\begin{remark}
    One can observe that the atoms of~$(\mathscr{B}(\tuple{a}), \preceq)$ are precisely the elements of~$\Delta_s(n)$ that cover $\tuple{a}$ and that 
    \begin{equation*}
        0_{\mathscr{B}(\tuple{a})}=\tuple{a}\qquad\textup{ and }\qquad 1_{\mathscr{B}(\tuple{a})}=\tuple{a}+e^{(m-1)}-e^{(M)}
    \end{equation*}
    where $m=\min\suppH((a_1,\ldots,a_{s-1}))$ and $M=\max\suppH((a_1,\ldots,a_{s-1}))$.
\end{remark}

We derive an explicit formula for the Möbius function of~$\Delta_s(n)$.  

\begin{theorem}
    The Möbius function of the interval $[\tuple{a}, \tuple{b}] \subseteq \Delta_s(n)$ is
    \begin{equation*}
        \mu_s(n; \tuple{a}, \tuple{b}) =
        \begin{cases}
            (-1)^{\sum_{j=0}^{s-1} (b_j - a_j)} & \textup{if } \tuple{b} \in \mathscr{B}(\tuple{a}),\\[4pt]
            0 & \textup{otherwise.}
        \end{cases}
    \end{equation*}
\end{theorem}

\begin{proof} 
By Proposition~\ref{prop: bool}, $(\mathscr B(\tuple a),\preceq)$ is a Boolean lattice
whose atoms above $\tuple a$ are precisely the elements
\[
\tuple a+(e^{(j-1)}-e^{(j)}),
\qquad
j\in\suppH((a_1,\ldots,a_{s-1})).
\]
Hence, for every $\tuple b\in\mathscr B(\tuple a)$ there exists a unique
$S\subseteq\suppH(((a_1,\ldots,a_{s-1}))$ such that
\[
\tuple b-\tuple a=\sum_{j\in S}(e^{(j-1)}-e^{(j)}).
\]
The rank of $\tuple b$ above $\tuple a$ in this Boolean lattice is $|S|$, and this rank coincides with $\widehat b_{s-1}-\widehat a_{s-1}$ in $\Delta_s(n)$. Therefore the Möbius function satisfies $\mu(\tuple a,\tuple b)=(-1)^{|S|}$. This proves the stated formula when $\tuple b\in\mathscr B(\tuple a)$. Suppose now that $\tuple b\notin\mathscr B(\tuple a)$.
Consider the Boolean sublattice
\[
\mathscr B(\tuple a)\cap[\tuple a,\tuple b]\subseteq[\tuple a,\tuple b].
\]
This is a proper Boolean sublattice whose maximal element is strictly smaller than
$\tuple b$.
By the defining recursion of the Möbius function,
\[
\sum_{\tuple a\preceq \tuple x\preceq \tuple b}\mu(\tuple a,\tuple x)=0,
\qquad \tuple a\neq\tuple b.
\]
The sum of the Möbius values over the Boolean sublattice
$\mathscr B(\tuple a)\cap[\tuple a,\tuple b]$ vanishes by the alternating-sum
property of Boolean lattices.
All remaining elements $\tuple x$ in the interval satisfy
$\tuple x\succcurlyeq\max(\mathscr B(\tuple a)\cap[\tuple a,\tuple b])$.
Consequently, the above recursion forces $\mu(\tuple a,\tuple b)=0$. This conlcudes the proof. \qedhere

\end{proof}
The explicit expression for the Möbius function provides additional structural information on the lattice $\Delta_s(n)$. 
In particular, it allows us to describe the irreducible elements and the length of maximal chains. Recall that in any finite lattice, an element $x$ is said to be \emph{join-irreducible} if it covers exactly one element, 
and \emph{meet-irreducible} if it is covered by exactly one element.  
As established in Theorem~\ref{thm:cover}, the covering relations in $\Delta_s(n)$ are determined by transfers of a single unit between two consecutive coordinates.

\begin{corollary}
    In $\Delta_s(n)$, the join-irreducible elements are precisely those of the form
    \[
        (0,\ldots,0,\underbrace{n-k}_{t\text{-th}},0,\ldots,0,k),
    \]
    and the meet-irreducible elements are precisely those of the form
    \[
        (k,0,\ldots,0,\underbrace{n-k}_{t\text{-th}},0,\ldots,0),
    \]
    for some $k \in \{0, \ldots, n - 1\}$ and $t\in\{0,\ldots,s-2\}$.
\end{corollary}
\begin{proof}
By Theorem~\ref{thm:cover}, an element
    $\tuple{a} = (a_0, \ldots, a_{s-1}) \in \Delta_s(n)$ covers another element
    $\tuple{b} = (b_0, \ldots, b_{s-1}) \in \Delta_s(n)$
    if and only if there exists a unique index $j \in \{0, \ldots, s-2\}$ such that
    $b_j = a_j - 1$, $b_{j+1} = a_{j+1} + 1$, and $b_i = a_i$ for all other indices $i$. Hence, $\tuple{a}$ covers exactly one element if and only if there is exactly one
    index $j \in \{0, \ldots, s-2\}$ for which $a_j \neq 0$.
    Writing $a_j = n - k$ for some $k$, we obtain
    \[
        \tuple{a} = (0, 0, \ldots, 0, n - k, 0, \ldots, 0, k),
    \]
    which is therefore join-irreducible. By symmetry of the dominance order, the meet-irreducible elements are obtained by
    reversing the coordinates, yielding
    \[
        (k, 0, \ldots, 0, n - k, 0, \ldots, 0).
    \]
    These are exactly the elements covered by a unique element in $\Delta_s(n)$.
\end{proof}

We recall the definition and, as corollary of Theorem~\ref{thm:cover} prove that all saturated chains in $\Delta_s(n)$ have the same length, which shows that the lattice of weak compositions is graded. 
\begin{definition}
    A \emph{saturated chain} in $\Delta_s(n)$ is a sequence
    \[
        \boldsymbol{0} = \tuple{a}^{(0)} \cov \tuple{a}^{(1)} \cov \cdots \cov \tuple{a}^{(t)} = \boldsymbol{1},
    \]
    where each $\tuple{a}^{(i)}$ covers $\tuple{a}^{(i-1)}$, and $\boldsymbol{0}$ and $\boldsymbol{1}$ denote, respectively, the minimal and maximal elements of $\Delta_s(n)$. The \emph{length} of a saturated chain is the number of covering relations it contains.
\end{definition}

\begin{corollary}
    Every saturated chain in $\Delta_s(n)$ from $\boldsymbol{0}$ to $\boldsymbol{1}$ has length exactly $sn$.
\end{corollary}

\begin{proof}
    By Theorem~\ref{thm:cover}, each covering relation transfers a single unit from position $j+1$ to position $j$ for some $j \in \{0, \ldots, s-2\}$. Starting from the minimal element $\boldsymbol{0} = (0, \ldots, 0, n)$, all $n$ units are initially located in the last coordinate. To reach the maximal element $\boldsymbol{1} = (n, 0, \ldots, 0)$, each of the $n$ units must successively move across all $s-1$ intermediate coordinates. Since each such move corresponds to one covering relation, and each unit requires $s$ transfers in total, the resulting saturated chain consists of exactly $sn$ covering relations.
\end{proof}

We conclude this section by showing the existence of an anti-isomorphism in $\Delta_s(n)$.

\begin{proposition}\label{prop:antiiso}
    The following map  is an anti-isomorphism in $\Delta_s(n)$.
    $$\varphi : \Delta_s(n)\to\Delta_s(n):(a_0,a_1,\ldots,a_{s-1})\mapsto (a_{s-1}, a_{s-2}, \ldots, a_1, a_0)$$
\end{proposition}

\begin{proof}
    The map $\varphi$ clearly preserves the sum of the components and is bijective and an involution. It therefore remains to show that $\varphi$ reverses the dominance order. Let $\tuple{a} = (a_0, \ldots, a_{s-1})$ and $\tuple{b} = (b_0, \ldots, b_{s-1})$ be elements of $\Delta_s(n)$. For $\varphi(\tuple{a}) = (a_{s-1}, \ldots, a_0)$, the associated sequence satisfies
    \[
        \widehat{\varphi(\tuple{a})}_j
        = \sum_{i=0}^j a_{s-1-i}
        = \widehat{a}_{\,s-1} - \widehat{a}_{\,s-2-j}
        \qquad\text{for all } j\in\{0,\ldots,s-1\},
    \]
    where we adopt the convention $\widehat{a}_{-1} = 0$. An analogous identity holds for $\widehat{\varphi(\tuple{b})}_j$. Consequently, we have $ \widehat{\varphi(\tuple{b})}_j \le \widehat{\varphi(\tuple{a})}_j$ if and only if $
        \widehat{b}_{\,s-2-j} \le \widehat{a}_{\,s-2-j}.$ As $j$ ranges over the set $\{0, \ldots, s-1\}$, so does $s-2-j$, and therefore $ \widehat{a}_k \le \widehat{b}_k$, for all $k\in\{0, \ldots, s-1\}$, if and only if $\widehat{\varphi(\tuple{b})}_j \le \widehat{\varphi(\tuple{a})}_j$, for all $j\in\{0, \ldots, s-1\}$. By Remark~\ref{rem:brylawski2}, this equivalence yields
    \[
        \tuple{a} \preceq \tuple{b} \iff \varphi(\tuple{b}) \preceq \varphi(\tuple{a}),
    \]
    which proves that $\varphi$ reverses the dominance order. Hence, $\varphi$ is an anti-isomorphism in $\Delta_s(n)$.
\end{proof}

\section{Codes and Anticodes}\label{sec:anticode}
The study of the combinatorial properties of the lattice $\Delta_{s+1}(n)$ in the previous section was motivated by its connection with submodules of $\Ring^n$, where $\Ring$ is a finite chain ring of nilpotency index $s$. Indeed, each element of $\Delta_{s+1}(n)$ encodes the distribution of coordinates among successive powers of the maximal ideal $\langle \gamma \rangle$, and the dominance order on~$\Delta_{s+1}(n)$ corresponds precisely to inclusion among the associated submodules. In coding theory, these submodules arise naturally as \textit{ring-linear anticodes}.
Thus, the lattice $\Delta_{s+1}(n)$ provides a combinatorial model for the inclusion relations among such anticodes, allowing us to describe their structure using the Möbius function and the covering relations derived in the previous section. 

\subsection{Ring-Linear Codes}
In this section, we recall some basic definitions and results on linear codes over finite chain rings. We restrict our attention to commutative finite chain rings, for which right and left ideals coincide. Throughout, let $\Ring$ denote  a finite chain ring with maximal ideal generated by $\gamma$, let $s$ be its nilpotency index and $q=p^r$ be the size of the residue field, for a positive integer $r$.

\begin{definition}
   A \emph{$\Ring$-linear code of length} $n$ is an $\Ring$-submodule~$\code \subseteq \Ring^n$. The elements of~$\code$ are called \emph{codewords}.
\end{definition}

By the fundamental theorem of finite abelian groups, each linear code $\code$ over a finite chain ring $\Ring$ is isomorphic to the following direct sum of $\Ring$-modules
\begin{align}\label{equ:directSum_code}
    \code \cong \bigoplus_{i = 0}^{s-1} (\Ring / \gamma^{s-i} \Ring)^{k_i}.
\end{align}
We refer to the the unique tuple $(k_0, \ldots, k_{s-1})$ as the \emph{subtype} of the code $\code$. 

\begin{remark}\label{rem:type}
    The subtype relates to the \emph{type} of an $\Ring$-module, which is defined as the partition $\tuple{\lambda} = (s^{k_0}, (s-1)^{k_1}, \ldots, 1^{k_{s-1}} )$.    
\end{remark}

As in classical coding theory over finite fields, we define the $\Ring$-dimension $k$ of $\code$ by $k = \log_{|\Ring|}(|\code|)$. Using the unique decomposition of $\code$ given in \eqref{equ:directSum_code}, we obtain
\begin{align}\label{eq:dimension_code}
    k := \sum_{i = 0}^{s-1} \frac{s-i}{s} k_i .
\end{align}
We recall that a linear code $\code \subseteq \Ring^n$ can be represented by a set of codewords that generate $\code$. We refer to such a set as a \emph{generating set}. The \emph{rank} $K$ of $\code$ is the cardinality of a minimal generating set (where minimal is understood with respect to inclusion), and it can be computed as
\begin{equation*}
    \label{equ:rank_code}
    K := \sum_{i = 0}^{s-1} k_i .
\end{equation*}
Note that the rank $K$ coincides with the $\Ring$-dimension $k$ if and only if $K = k = k_0$. In this case, by \eqref{equ:directSum_code}, a code of subtype $(k_0, 0, \ldots, 0)$ is isomorphic to $\code \cong \Ring^{k_0}$ and therefore admits an $\Ring$-basis. In the following, we denote by $\code \subseteq \Ring^n$ a linear code of rank $K$ and subtype $(k_0, \ldots, k_{s-1})$.

\begin{definition}
    We refer to $k_0$ as the \emph{free rank} of $\code$  and we denote it by $\text{freerk}(\code)$. Moreover, the code $\code$ is said to be \emph{free} if $K = k_0$. 
\end{definition}

We define the dual code using the standard inner product as
\begin{equation*}
    \mathcal{C}^\perp =\{ x \in \Ring^n \st \langle x,c \rangle =0 \ \forall c \in \mathcal{C}\}.
\end{equation*}

We have that $\mathcal{C}^\perp$ has subtype $(n - K, k_{s-1}, \ldots, k_1)$ and rank $n - k_0$. As in the case of codes over finite fields, codes over $\Ring$ can also be represented using matrices. 

\begin{definition}\label{def:GH}
        A matrix $G \in \Ring^{K \times n}$ is a \textit{generator matrix} of $\mathcal{C}$ if the rows of $G$ span the code. A \textit{parity-check matrix} $H$ is an $(n-k_0)\times n$ matrix over $\Ring$ whose null-space coincides with $\mathcal{C}$.
\end{definition}

It is often helpful to consider these matrices in the systematic form introduced in~\cite{norton2000structure}.  In the following, we say that the codes $\code$ and $\mathcal{D}$ are
\emph{permutation equivalent} if there exists a permutation matrix $P$ such that $\mathcal{D}=\code P$.

\begin{proposition}[\text{\cite[Proposition 3.2]{norton2000structure}}]\label{prop:sysformG}
     We have that $\code$ is permutation equivalent to a code that admits a generator matrix in the following systematic form.
    \begin{align}\label{equ:systematicformG}
        G =
        \begin{bmatrix}
            I_{k_0}&A_{0,1} &A_{0,2} &A_{0,3}&\dots& A_{0,s-1}& A_{0,s} \\
            0 &\gamma I_{k_1} & \gamma A_{1,2}& \gamma A_{1,3}& \dots& \gamma A_{1,s-1}&\gamma A_{1,s}\\
            0 &0 & \gamma^2 I_{k_2} & \gamma^2 A_{2,3}& \dots& \gamma^2 A_{2,s-1}&\gamma^2 A_{2,s}\\
            \vdots &  \vdots& \vdots & \vdots && \vdots &\vdots \\
            0 &0 & 0 & 0 & \dots & \gamma^{s-1}I_{k_{s-1}}& \gamma^{s-1} A_{s-1,s}\\
        \end{bmatrix}\ ,
    \end{align} 
    where $A_{i,s}\in (\Ring / \gamma^{s-i} \Ring)^{k_{i}\times (n-K)}$ and $A_{i,j}\in (\Ring / \gamma^{s-i} \Ring)^{k_i\times k_j}$ for $j< s$.
\end{proposition}

Similarly to the subtype of a code $\code \subseteq \Ring^n$, which is defined in terms of the ideals corresponding to the rows of a generator matrix in systematic form, there also exists an analogous parameter for the columns, that is, for the positions of each codeword in $\code$.

\begin{definition}\label{def:support_subtype}
    For each $j \in \set{1, \ldots, n}$, we define the $j$-th coordinate map by
    \begin{align}
        \begin{array}{rccc}
            \pi_j : & \Ring^n & \longrightarrow & \Ring, \\
                & (r_1, \ldots, r_n) & \longmapsto & r_j.
        \end{array}
    \end{align}
    We also define the number of coordinates $j \in \set{1, \ldots, n}$ that generate the ideal $\ideal{\gamma^i}$ by
    \begin{align}
        n_i(\code) := \card{\set{j \in \set{1, \ldots, n} \st \ideal{\gamma^i} = \ideal{\pi_j(\code)}}}.
    \end{align}
    We refer to the $(s+1)$-tuple $(n_0(\code), \ldots, n_s(\code))$ as the \emph{support subtype} of $\code$.
\end{definition}
Note that $n_s(\mathcal{C})$ denotes the number of positions $i$, such that for all codewords $c \in \code$ we have $c_i=0$. Thus, if $n_s=0$, the code is non-degenerate.\medbreak 

Several metrics are considered for ring-linear codes, among which the most prominent are the Hamming, Lee, and homogeneous metrics.

\begin{definition}
    Let $x, y \in \Ring^n$. The \emph{Hamming weight} of $x$, denoted by $\wtH(x)$, is defined as the number of nonzero entries in $x$, that is,
\begin{align}
\wtH(x) := |\{i \in \{1, \ldots, n\} \mid x_i \neq 0 \}|.
\end{align}
This induces the \emph{Hamming distance} between $x$ and $y$, denoted by $\HD(x,y)$, which is given by the number of coordinates in which $x$ and $y$ differ, namely $\HD(x,y):=\wtH(x-y)$.
\end{definition}

The homogeneous weight is finer, as it distinguishes nonzero elements lying in the socle $\gamma^{s-1}\Ring$, also denoted by $\langle \gamma^{s-1} \rangle$, which is of size $p$.

\begin{definition}
    The \emph{normalized homogeneous weight} $\HomW(x)$ of an element $x\in \Ring$ is 
        $$\HomW(x):=\begin{cases} 0 & \text{ if } x=0, \\ 1 & \text{ if } x \in \Ring \setminus \langle \gamma^{s-1} \rangle, \\ \frac{p}{p-1} & \text{ if } x \in \langle \gamma^{s-1} \rangle \setminus \{0\}. \end{cases}$$
    For tuples, $x \in \Ring^n$, we extend the homogeneous weight additively, coordinate-wise, by
    \begin{align}
        \HomW(x):= \sum_{i=1}^n \HomW(x_i).
    \end{align}
    The homogeneous weight naturally induces the \emph{homogeneous distance} between two elements $x, y \in \Ring^n$, defined by $\HomD(x,y):=\HomW(x-y).$
\end{definition}

Another metric that we consider is the Lee metric, introduced in~\cite{lee1958some}. Note that the Lee metric is defined only over integer residue rings, that is, $\mathbb{Z}/m\mathbb{Z}$ for any positive integer $m$. Since we restrict our attention to finite chain rings, we henceforth take $\Ring = \mathbb{Z}/p^s\mathbb{Z}$, where $p$ is a prime and $s$ is a positive integer.

\begin{definition}
    The \emph{Lee weight} $\LW(x)$ of an element $x \in \mathbb{Z}/p^s\mathbb{Z}$ is defined by
\begin{align}
\LW(x) := \min{x, |p^s - x|}.
\end{align}
The Lee weight is extended additively to vectors: for $x \in (\mathbb{Z}/p^s\mathbb{Z})^n$, we set
\begin{align}
\LW(x) := \sum_{i=1}^n \LW(x_i).
\end{align}
For ant $x, y \in (\mathbb{Z}/p^s\mathbb{Z})^n$, the \emph{Lee distance} between $x$ and $y$ is $\LD(x,y):=\LW(x-y).$
\end{definition}
Note that for $p^s \in \{2,3\}$ the Lee weight coincides with the Hamming weight. Therefore, whenever we refer to the Lee weight, we implicitly assume that the underlying integer residue ring is different from $\mathbb{Z}/2\mathbb{Z}$ and $\mathbb{Z}/3\mathbb{Z}$. Given a ring-linear code $\code \subseteq \Ring^n$, and any of the three weights defined above, now denoted generically by $\wt$, we define the \emph{minimum distance} of $\code$ as $ \dist(\mathcal{C}) = \min \{\text{wt}(c) \st c \in \mathcal{C} \setminus \{0\}\}$.

\subsection{Ring-linear Anticodes}
In this section, we introduce the notion of optimal ring-linear anticodes, which naturally extends the definition given in~\cite[Definition~7]{ravagnani2016generalized} to the setting of ring-linear codes. We provide a characterization of these anticodes, with particular emphasis on their representation with respect to the Lee metric. Finally, we establish a connection with the lattice of weak compositions introduced in Section~\ref{sec:new-lattices}. In particular, we show that the support subtypes partition the family of anticodes, yielding a natural link with this lattice structure, and we derive  relationships between these objects.

\begin{definition}
    Let $r$ be a nonnegative integer. The \emph{maximum weight} of $\code$ is $\text{maxwt}(\code)=\max\{\wt(c) \mid c \in \code\}$. An $r$-\emph{anticode} $\mathcal{A}$ is a $\Ring$-linear code of length $n$ with maximum weight at most $r$, that is $\text{maxwt}(\code)\leq r$.
\end{definition}
 
The classical code–anticode bound still holds over $\Ring$, as stated below. The proof is similar to the classical case and is therefore omitted.

\begin{proposition}
Let $\mathcal{C} \subseteq \Ring^n$ be a linear code of $\Ring$-dimension $k$ and minimum  distance~$d$, and $\mathcal{A}$ be a  $\Ring$-linear $(d-1)$-anticode of $\Ring$-dimension $k'$, then $k+k' \leq n$.
\end{proposition}

To generalize the classical anticode bound to rings, we must distinguish which weight is being considered. 
\begin{proposition}\label{prop:AnticodeBound-Hamming} 
    Let $\mathcal{A} \subseteq \Ring^n$ be a linear code of rank $K$ and maximum Hamming weight $w$. We have $w \geq K.$
\end{proposition}

\begin{proof}
   Let  $(k_0, \ldots, k_{s-1})$ denote the subtype of $\mathcal{A}$ with $\sum_{i=0}^{s-1} k_i = K$, and consider a generator matrix of the form
   \begin{align*}
       G = 
       \begin{pmatrix} 
           I_{k_0} & A_{0,1} & \cdots & A_{0,s-1} & 0 \\ 
           0 & \gamma I_{k_1} & \cdots & \gamma A_{1,s-1} & 0 \\ 
           \vdots & \vdots & \ddots & \vdots & \vdots \\ 
           0 & 0 & \cdots & \gamma^{s-1} I_{k_{s-1}} & 0
       \end{pmatrix}.
   \end{align*}
   By the row-reduced form, we may assume that $\gamma^i A_{i,j}$ has no entries in $\langle \gamma^j \rangle$ for all $i \in \{0, \ldots, s-2\}$ and $i < j \leq s-1$. Summing over all rows of $G$, we can then construct a codeword $c$ with Hamming weight $\wtH(c) = K$. Note that if the last $n-K$ columns are not all zero, the rank remains $K$, but the maximum Hamming weight increases.
\end{proof}

We say that a code $\mathcal{A} \subseteq \Ring^n$ is an \emph{optimal Hamming-metric anticode} if it meets the bound in Proposition~\ref{prop:AnticodeBound-Hamming} with equality, that is, if $\text{maxwt}^\text{H}(\mathcal{A}) = K$. The following example illustrates an optimal Hamming-metric anticode over the ring $\mathbb{Z}/9\mathbb{Z}$.

\begin{example}\label{ex:optimal_hamming}
    Let $\mathcal{A} \subseteq (\mathbb{Z}/9\mathbb{Z})^3$ be the code of rank $K = 2$ generated by
    $$G= \begin{pmatrix} 1 & 2 & 0 \\ 0 & 3 & 0 \end{pmatrix}.$$ One can easily check that the maximum Hamming weight of $\mathcal{A}$ is $w = 2 = K$. Hence, $\mathcal{A}$ is an optimal Hamming-metric anticode over $\mathbb{Z}/9\mathbb{Z}$.
\end{example}

For the homogeneous metric, we note that the maximum homogeneous weight is always attained in the socle $\mathcal{A} \cap \langle \gamma^{s-1} \rangle^n$. Therefore, the anticode bound in the homogeneous metric can be stated as follows.

\begin{proposition}\label{prop:AnticodeBound-Hom} 
    Let $\mathcal{A} \subseteq \Ring^n$ be a linear code of rank $K$ and maximal homogeneous weight $w$. Then $w \geq K \frac{p}{p-1}.$
\end{proposition}

\begin{proof}
    For some $i \in \{0, \ldots, s-1\}$, consider an element $a \in \langle \gamma^i \rangle \setminus \langle \gamma^{i+1} \rangle$. Note that $a \gamma^{s-1-i} \in \langle \gamma^{s-1} \rangle \setminus \{0\}$, and hence 
    \begin{align}\label{eq:hom} 
        \HomW(a \gamma^{s-i-1}) = \frac{p}{p-1} \geq \HomW(a).
    \end{align}
    Thus, for any subtype $(k_0, \ldots, k_{s-1})$ of $\mathcal{A}$ with $\sum_{i=0}^{s-1} k_i = K$ and a generator matrix of the form
    \begin{align}
        G = 
        \begin{pmatrix} 
            I_{k_0} & A_{0,1} & \cdots & A_{0,s-1} & 0 \\ 
            0 & \gamma I_{k_1} & \cdots & \gamma A_{1,s-1} & 0 \\ 
            \vdots & \vdots & \ddots & \vdots & \vdots \\ 
            0 & 0 & \cdots & \gamma^{s-1} I_{k_{s-1}} & 0
        \end{pmatrix},
    \end{align}
    we obtain a generator matrix for the socle $\mathcal{A} \cap \langle \gamma^{s-1} \rangle^n$ of the form
    $$ \begin{pmatrix} \gamma^{s-1} I_K & 0 \end{pmatrix}.$$ Hence, in any code of rank $K$, we can construct a codeword $c$ with homogeneous weight $\HomW(c) = K \frac{p}{p-1}$. For all other codewords, Equation~\eqref{eq:hom} implies that their homogeneous weight does not exceed $K \frac{p}{p-1}$. Finally, if the last $n-K$ columns of $G$ are not all zero, the rank remains $K$, but the maximum homogeneous weight increases.
\end{proof}

We say that a code $\mathcal{A} \subseteq \Ring^n$ is an \emph{optimal homogeneous-metric anticode} if it meets the bound in Proposition~\ref{prop:AnticodeBound-Hom} with equality. Since we can identify $\langle \gamma^{s-1} \rangle$ with $\mathbb{F}_p$, the optimal anticodes in the homogeneous metric are obtained by lifting optimal anticodes in the Hamming metric. In particular, Example~\ref{ex:optimal_hamming} also provides an example of an optimal homogeneous-metric anticode.

\begin{example}\label{ex1}
   Let $\mathcal{A}$ over $\mathbb{Z}/9\mathbb{Z}$ be the code generated by
    $$G= \begin{pmatrix} 1 & 2 & 0 \\ 0 & 3 & 0 \end{pmatrix}.$$ The codeword $(3,3,0)$ has maximal homogeneous weight $w = 3 = K \frac{p}{p-1}$. Hence, $\mathcal{A}$ is an optimal homogeneous-metric anticode over $\mathbb{Z}/9\mathbb{Z}$.
\end{example}

The most interesting case arises with the Lee metric. Indeed, in this case we get a full characterization of the optimal anticodes.

We begin by stating the Lee-metric anticode bound over $\mathbb{Z}/p^s\mathbb{Z}$.

\begin{proposition}\label{prop:anticodeBound-Lee}
    Let $p \neq 2$ and let $\mathcal{A} \subseteq (\mathbb{Z}/p^s\mathbb{Z})^n$ be a linear code of subtype $(k_0, \ldots, k_{s-1})$ and maximum Lee weight $w$. For each $i \in \{0, \ldots, s-1\}$, define the maximum Lee weight attained by elements of the ideal $\langle p^i \rangle$ as $M_i := \max \{ \LW(\lambda) \st \lambda \in \langle p^i \rangle\} =\frac{p^{s}+p^{i}}{2}$. We have $w \geq \sum_{i=0}^{s-1} k_i M_i$.
\end{proposition}
\begin{proof}
As in the proofs of Propositions~\ref{prop:AnticodeBound-Hamming} and~\ref{prop:AnticodeBound-Hom}, We assume that the last $n-K$ columns are zero, as they would only increase the maximum weight, and we are interested in minimizing the maximal Lee weight. Let $G$ be a generator matrix of $\mathcal{A}$ of the form
\begin{align}
    G = 
    \begin{pmatrix} 
        I_{k_0} & A_{0,1} & \cdots & A_{0,s-1} & 0 \\ 
        0 & p I_{k_1} & \cdots & p A_{1,s-1} & 0 \\ 
        \vdots & \vdots & \ddots & \vdots & \vdots \\ 
        0 & 0 & \cdots & p^{s-1} I_{k_{s-1}} & 0
    \end{pmatrix}.
\end{align}
Since every codeword $c \in \mathcal{A}$ can be written as $c = mG$ for some $m \in (\mathbb{Z}/p^s\mathbb{Z})^K$, the entries of $c$ appear in blocks of the form $m A_j$ for some $j \in \{0, \ldots, s-1\}$, that is
\begin{align}
    A_j := 
    \begin{pmatrix} 
        A_{0,j} \\ p A_{1,j} \\ \vdots \\ p^{j-1} A_{j-1,j}e \\ p^j I_{k_j} 
    \end{pmatrix}.
\end{align}
The maximal Lee weight among these entries is therefore the maximum Lee weight attained in a coset of $\langle p^j \rangle$. Hence, for $a \in \mathbb{Z}/p^s\mathbb{Z}$ and $i \in \{0, \ldots, s-1\}$, we define the maximal coset Lee weight as $$N_i(a) := \max\{\LW(a+ b p^i) \mid  b \in \{0, \ldots, p-1\}\}. $$

Observe that, for all $i \in \{0, \ldots, s-1\}$ and any $a \in \langle p^i \rangle$, we have $N_i(a) = M_i$, since~$M_i \geq M_{i+1}$. In particular, the maximum Lee weight in $\langle p^i \rangle$ is attained only by the elements $p^i \lambda$, for $\lambda = -1/2$, and by $p^i(\lambda+1) = -p^i \lambda$. On the other hand, if $a \notin \langle p^i \rangle$, the coset $\{a, a+p^i, \ldots, a+p^i (p-1)\}$ does not contain either $p^i \lambda$ or $-p^i \lambda$, which are the only elements attaining the Lee weight $M_i$. Nevertheless, we can find an element~$\mu \in \mathbb{Z}/p^s\mathbb{Z}$ such that $p^i \lambda +a \leq \mu \leq p^i(\lambda+1)+a$. By definition of the Lee weight, it follows that $\LW(\mu) \geq \LW(p^i \lambda)$. Therefore, we can construct a codeword with $\LW(c) \geq \sum_{i=0}^{s-1} M_ik_i$. Moreover, for a generator matrix $G$ with all $A_{i,j} = 0$, there exists a codeword $c$ satisfying~$\LW(c)=\sum_{i=0}^{s-1} k_i M_i.$ This concludes the proof.
\end{proof}

We say that a code $\mathcal{A} \subseteq \Ring^n$ is an \emph{optimal Lee-metric anticode} if it meets the bound in Proposition~\ref{prop:anticodeBound-Lee} with equality. Note that, in the proof of Proposition~\ref{prop:anticodeBound-Lee}, we showed that the following matrix generates an optimal Lee-metric anticode:

\begin{equation}\label{genAk}
    A_{(k_0, \ldots, k_{s-1},n-K)}= 
    \begin{pmatrix}
        I_{k_0}  & 0  & \cdots & 0 & 0  \\
        0 & p I_{k_1} &   \cdots & 0 & 0  \\
        \vdots & \vdots &   & \vdots & \vdots \\
        0 & 0   & \cdots &  p^{s-1} I_{k_{s-1}} & 0
    \end{pmatrix}.
\end{equation}

Codes generated by matrices of the form $A_{(k_0, \ldots, k_{s-1}, n-K)} \in (\mathbb{Z}/p^s\mathbb{Z})^{K \times n}$ are degenerate, unless $n = K$, in which case the code is trivial. We have restricted our attention to the case $p \neq 2$, since for $p = 2$ one may also encounter optimal anticodes in the classical setting that are non-degenerate. For example, one can verify that the code generated by 
$$G= \begin{pmatrix} 1 & 0 & 1 \\ 0 & 1 & 1 \end{pmatrix}$$ is an optimal anticode over $\mathbb{F}_2$. While in the Hamming and homogeneous metrics we have seen examples of codes whose generator matrices do not necessarily have the form given in Equation~\eqref{genAk}, all optimal Lee-metric anticodes are generated by matrices of this form. The following result provides a characterization of the optimal Lee-metric anticodes.

\begin{theorem}\label{prop:antiunique} 
    If $\mathcal{A}$ is permutation equivalent to a code generated by $A_{(k_0, \ldots, k_{s-1}, n-K)}$, then $\mathcal{A}$ is an optimal Lee-metric anticode. Conversely, if $p \neq 2$ and $\mathcal{C}$ is an optimal Lee-metric anticode of length $n$ and subtype $(k_0, \ldots, k_{s-1})$, then $\mathcal{A}$ is permutation equivalent to a code generated by $A_{(k_0, \ldots, k_{s-1}, n-K)}$.
\end{theorem}
\begin{proof}
We already observed in the proof of Proposition \ref{prop:anticodeBound-Lee} that for any $a \in \mathbb{Z}/p^s \mathbb{Z}$, if $$A_j'= \begin{pmatrix} A_{0,j} \\ p A_{1,j} \\ \vdots \\ p^{j-1} A_{j-1,j} \end{pmatrix}$$ 
is not the  zero matrix for all $j \in \{0, \ldots, s-1\}$, then we must consider the Lee weights of the cosets $a + \langle p^j \rangle$. In order for all codewords $c$ to satisfy 
$$\LW(c) \leq \sum_{i=0}^{s-1} k_i M_i$$ these coset Lee weights must fulfill $M_i \leq \LW(a+ \lambda p^i) \leq M_i = \LW( \lambda p^i), $ which forces $a=0$. Consequently, we must have $mA_j' = 0$ for all $m \in (\mathbb{Z}/p^s\mathbb{Z})^{k_0 + \cdots + k_{j-1}}$, implying that $A_j' = 0$. Therefore, we can restrict our attention to a generator matrix of the form
 $$G= \begin{pmatrix} I_{k_0} & 0 & \cdots & 0 & A_{0,s} \\ 0 & p I_{k_1} & \cdots &  0 & p A_{1,s} \\
 \vdots & \vdots & \ddots & \vdots & \vdots \\ 
 0 & 0 & \cdots & p^{s-1} I_{k_{s-1}} & p^{s-1} A_{s-1,s}
 \end{pmatrix}.$$

It remains to show that the last $n-K$ columns must be zero columns. We define $$A_s= \begin{pmatrix} A_{0,s} \\ p A_{1,s} \\ \vdots \\ p^{s-1} A_{s-1,s} \end{pmatrix}.$$

 Any codeword $c$ is of the form $mG$, and by choosing $m \in \{ \pm \lambda\}^K$, we obtain codewords of the form
 $$c= (\underbrace{\pm \lambda, \ldots,  \pm \lambda}_{k_0}, \ldots, \underbrace{\pm \lambda p^{s-1}, \ldots, \pm \lambda p^{s-1}}_{k_{s-1}}, b),$$ where $b = mA_s$.
Since $\mathcal{A}$ is an optimal Lee-metric anticode, we have $$\sum_{i=0}^{s-1} k_i M_i \leq \LW(c) = \sum_{i=0}^{s-1} k_iM_i + \LW(b) \leq \sum_{i=0}^{s-1} k_iM_i,$$ and hence all elements $b = mA_s$ must be zero. As this holds for all $m \in \{ \pm \lambda\}^K$, it follows that $A_s = 0.$
\end{proof}

The following example is in line with Theorem~\ref{prop:antiunique}

 \begin{example}
    Let $\mathcal{A} \subseteq \mathbb{Z}/9\mathbb{Z}^3$ be the code from Example~\ref{ex1}, that is, the code generated by the matrix 
    $$G= \begin{pmatrix} 1 & 2 & 0 \\ 0& 3 & 0 \end{pmatrix}$$ 
    with $M_0 = 4$ and $M_1 = 3$. We have $k_0 = k_1 = 1$, and we can construct a codeword $(4,5,0)$ with Lee weight $8$. Since
    $$ 8 > 7 = 4+3 = \sum_{i=0}^{s-1} k_iM_i$$ the code is \emph{not} an optimal Lee-metric code. On the other hand, consider the code over $\mathbb{Z}/9\mathbb{Z}$ generated by
    $$G'= \begin{pmatrix}
         1 & 0 & 0 \\ 0 & 3 & 0 
    \end{pmatrix}.$$ 
    This code is an optimal Lee-metric code, as the maximum Lee weight attained is $7$, for example by the codeword $(4,3,0)$.
\end{example}

\begin{remark}
To return to the finite chain rings, for a subtype $(a_0, \ldots, a_{s-1})$, we define 

\begin{equation}\label{genAk}
    A_{(a_0, \ldots, a_{s-1},a_s)}= 
    \begin{pmatrix}
        I_{a_0}  & 0  & \cdots & 0 & 0  \\
        0 & \gamma I_{a_1} &   \cdots & 0 & 0  \\
        \vdots & \vdots &   & \vdots & \vdots \\
        0 & 0   & \cdots &  \gamma^{s-1} I_{a_{s-1}} & 0
    \end{pmatrix},
\end{equation} and set $a_s=n- \sum_{i=0}^{s-1}a_i.$
For codes generated by $A_{(a_0, \ldots, a_s, a_s)}$, we observe that the subtype $(a_0, \ldots, a_{s-1})$ almost coincides with the support subtype $(a_0, \ldots, a_{s-1}, a_s)$, with the only difference being that $a_s = n - \sum_{i=0}^{s-1} a_i$ additionally accounts for the degeneracy of the code. Since non-trivial optimal Lee-metric anticodes are necessarily degenerate, the support subtype carries more relevant information.     
\end{remark}

In the remainder of this paper, we abuse terminology and refer to optimal Lee-metric anticodes simply as \emph{anticodes}, as these are the only anticodes of interest in this work. 
Let $\mathcal{A} \subseteq\Ring^n$ be an anticode with support subtype $\tuple{a} = (a_0, \ldots, a_{s})$. We have that the dual of $\mathcal{A}$ 
has support subtype $\tuple{a}^\perp= (n- \sum_{i=0}^{s-1} a_i =a_s, a_{s-1}, \ldots, a_0)$ and is generated by 
\begin{align}
    A_{(a_s, \ldots, a_0)} = 
    \begin{pmatrix} 
        0 & 0 &  \cdots & 0 & I_{a_s} \\ 
        0 &0 & \cdots &  \gamma I_{a_{s-1}} & 0 \\
        \vdots & \vdots & \ddots & \vdots & \vdots \\ 
        0 & \gamma^{s-1} I_{a_1} & \cdots & 0 & 0
    \end{pmatrix},
\end{align}
This shows a connection between the duality of anticodes and the anti-isomorphism described in Corollary~\ref{prop:antiiso}. We conclude this section by presenting a further connection between the lattice of weak compositions and families of anticodes.

\begin{proposition}\label{prop:order}
    Let $\tuple{a} = (a_0, \ldots, a_s)$ and $\tuple{b} = (b_0, \ldots, b_s)$ be the support subtypes of the anticodes $\mathcal{A} \subseteq \Ring^{n}$ and $\mathcal{B} \subseteq \Ring^{n}$, respectively. Let $\widehat{\tuple{a}}$ and $\widehat{\tuple{b}}$ denote the associated sequences of $\tuple{a}$ and $\tuple{b}$, respectively. Then $\mathcal{A} \subseteq \mathcal{B}$ if and only if
    \[
    \widehat{a}_j \leq \widehat{b}_j \quad\text{for each}  \quad j \in \{0, \ldots, s\}.
    \]
\end{proposition}

\begin{proof}
Assume first that $\widehat{a}_j \leq \widehat{b}_j$ for each $j \in \{0, \ldots, s\}$. By construction, the entries appearing in the $j$th coordinate of a codeword in $\mathcal{A}$ lie in an ideal determined by $\widehat{a}_j$, while those appearing in the $j$th coordinate of a codeword in $\mathcal{B}$ lie in a (possibly larger) ideal determined by $\widehat{b}_j$. The inequality $\widehat{a}_j \leq \widehat{b}_j$ therefore implies that, for each coordinate $j$, the set of possible entries of codewords in $\mathcal{A}$ is contained in the corresponding set for $\mathcal{B}$. Hence, $\mathcal{A} \subseteq \mathcal{B}$. Conversely, suppose that $\mathcal{A} \subseteq \mathcal{B}$. Then, for each coordinate, the ideal generated by the entries of $\mathcal{A}$ must be contained in the ideal generated by the entries of $\mathcal{B}$. Equivalently, for each column position, the exponent of $\gamma$ associated with $\mathcal{A}$ is greater than or equal to the corresponding exponent associated with $\mathcal{B}$ (with the convention that $0$ has exponent $s$). This condition is precisely equivalent to $\widehat{a}_j \leq \widehat{b}_j$ for every $j \in \{0, \ldots, s\}$, completing the proof.
\end{proof}

As a consequence of Propositions~\ref{prop:antiunique} and~\ref{prop:order}, the notion of support subtype of anticodes naturally induces an equivalence relation on the set 
of $\mathcal{R}$-anticodes. 
In particular, for any weak composition $\tuple{a} = (a_0, \ldots, a_s) \in \Delta_{s+1}(n)$, there exists a family of anticodes 
that are permutation equivalent to the anticode generated by $A_{(a_0, \ldots, a_{s-1})}$, and conversely, every such family arises in this way. We denote this family by $\mathcal{A}_{\tuple{a}}$, and we observe that $\rk(\mathcal{A}) = n - a_s$ for every $\mathcal{A} \in \mathcal{A}_{\tuple{a}}$. This introduces the following notation.

 \begin{definition}
    We denote by $\mathcal{A}(\Ring^n)$ the set of anticodes in $\Ring^n$, and by $\mathcal{A}_{\tuple{a}}(\Ring^n)$ the set of anticodes that are permutation equivalent to an anticode generated by $A_{\tuple{a}}$, for a weak composition $\tuple{a} \in \Delta_{s+1}(n)$.
\end{definition}

\section{Invariants}\label{sec:invariants}

In this section, we introduce and study some invariants associated with the anticodes introduced in Section \ref{sec:anticode}, namely the \emph{binomial moments} and the \emph{weight distribution}. Throughout the section,
we let $\Ring$ denote  a finite chain ring with maximal ideal generated by $\gamma$, let $s$ be its nilpotency index and $q=p^r$ be the size of the residue field, for a positive integer $r$
and we let $\code \subseteq \Ring^n$ be an $\Ring$-linear code of  rank $K$.  
We adopt the convention that $\sum_{x \in \emptyset} f(x) = 0$ for any function~$f : \mathbb{Q} \rightarrow \mathbb{Q}$. The next definitions are the analogue of~\cite[Definitions~6.5 and~6.6]{cotardo2022zeta} for ring-linear anticodes.

We recall that we are working with finitely generated submodules and introduce one of the main tools for this section. It is well known that the number of $k$-dimensional subspaces of an $n$-dimensional vector space over the finite field $\F_q$, where $q$ is a prime power, is given by the Gaussian coefficient
\begin{equation}\label{eq:gaussian}
    \gb{n}{k}_q
    = \frac{(q^n - 1)(q^{n-1} - 1)\cdots(q^{n-k+1} - 1)}
           {(q^k - 1)(q^{k-1} - 1)\cdots(q - 1)}=\prod_{i=0}^{k-1} \frac{(q^n-q^i)}{(q^k-q^i)}.
\end{equation}

A similar formula was established for the number of submodules of support subtype~$\tuple{b}$ of a finite module of support subtype $\tuple{a}$ over a finite chain ring  (see ~\cite{honold,macdonald1998symmetric} and ~\cite[Theorem~4]{georgieva2017representation}).

\begin{theorem}\label{thm:counting}
    Let $\tuple{a}=\subtype(\code)$. The number of subcodes $\mathcal{D}\subseteq \code$ with $\subtype(D)=\tuple{b}\preceq\tuple{a}$  is
    \begin{equation}\label{eq:counting}
      \gb{ \tuple{a} }{ \tuple{b} }_q= q^{\sum_{i=0}^{s-1} (\widehat{a_i}-\widehat{b_i}) \widehat{b}_{i-1} } \prod_{i=0}^{s-1} \gb{ \widehat{a_i} - \widehat{b}_{i-1}}{b_i}_q,       
          \end{equation}
        where $\Ring/ \langle \gamma \rangle= \mathbb{F}_q$.
\end{theorem}

\begin{example}
    Let $\Ring=\mathbb{Z}/9\mathbb{Z}$ and let $\code\subseteq \Ring^3$ be the code as given in Example~\ref{ex1}. Recall that $\code$ is generated by 
    $$G= \begin{pmatrix} 1 & 2 & 0 \\ 0 & 3 & 0 \end{pmatrix}.$$
    Since $\code=\Ring\oplus \Ring/(3\Ring)$, we get $\subtype(\code)=(1,1,1)$. We want to count the number of subcodes $\mathcal{D}\subseteq \code$ with $\subtype(\mathcal{D})=(0,1,2)$. There are $4$ such subcodes and they are the generated by the following matrices, respectively,
    \begin{equation*}
        \begin{pmatrix}
            3 & 0 & 0
        \end{pmatrix},\qquad
        \begin{pmatrix}
            0 & 3 & 0
        \end{pmatrix},\qquad
        \begin{pmatrix}
            3 & 3 & 0
        \end{pmatrix},\qquad
        \begin{pmatrix}
            3 & 6 & 0
        \end{pmatrix}.
    \end{equation*}
    By~\eqref{eq:counting} 
    we get
    \begin{equation*}
     \gb{(1,1,1)}{(0,1,2)}_3=3^{ \sum_{i=0}^1 (\widehat{a_i}- \widehat{b_i}  ) \widehat{b}_{i-1}}\gb{\widehat{a}_i - \widehat{b}_{i-1}}{b_i}_{3}
        = 3^{0(2-1)} \gb{1 - 0}{0}_{3}\gb{2- 0}{1}_{3}=4.
    \end{equation*}
\end{example}

\begin{definition}
     Let $\tuple{a}\in\Delta_{s+1}(n)$ and $j\in\{1,\ldots,K\}$. We define 
     \begin{equation*}
         W_\tuple{a}^{(j)}(\code)=\sum_{\mathcal{A}\in\mathcal{A}_\tuple{a}(\Ring^n)}W_\mathcal{A}^{(j)}(\code),
     \end{equation*}
     where $W_\mathcal{A}^{(j)}(\code)=|\{\mathcal{D}\leq \code\cap \mathcal{A} \mid \rk(\mathcal{D})=j \textup{ and }\mathcal{D}\leq B\leq \mathcal{A}, B\in\mathcal{A}(\Ring^n)\implies B=\mathcal{A} \}|$ for any $\mathcal{A}\in\mathcal{A}_\tuple{a}(\Ring^n)$. We refer to the set $\{W_\tuple{a}^{(j)}(\code) \mid \tuple{a}\in\Delta_{s+1}(n)\}$ as the $\Ring$-weight distribution of $\code$.
\end{definition}

\begin{definition}\label{def:bin}
    Let $\tuple{a}\in\Delta_{s+1}(n)$ and $j\in\{1,\ldots,K\}$. The $(\tuple{a},j)$-th $\Ring$-\textit{binomial moment} of $\code$ is
    \begin{equation*}
        B_{\tuple{a}}^{(j)}(\code)=\sum_{\mathcal{A}\in\mathcal{A}_\tuple{a}(\Ring^n)}B_\mathcal{A}^{(j)}(\code).
    \end{equation*}
    where $B_\mathcal{A}^{(j)}(\code)= |\{\mathcal{D}\leq \code\cap \mathcal{A} \mid \rk(\mathcal{D})=j\}|$ for any $\mathcal{A}\in\mathcal{A}_\tuple{a}(\Ring^n)(\Ring^n)$.
\end{definition}

\begin{remark}\label{rem:binom}
    One can easily check that as a consequence of Theorem~\ref{thm:counting}, we have 
    \begin{equation*}
        B_\mathcal{A}^{(j)}(\code)=\sum_{\substack{\tuple{b}\in\Delta_{s+1}(n)\\\tuple{b}\preceq\subtype(\code\cap \mathcal{A})\\
             b_s=n-j}}\gb{\subtype(\code\cap \mathcal{A})}{\tuple{b}}_p
    \end{equation*}
    for any $\mathcal{A}\in\mathcal{A}_\tuple{a}(\Ring^n)$, $\tuple{a}\in\Delta_{s+1}(n)$ and $j\in\{1,\ldots,K\}$.
\end{remark}

We illustrate this remark through the following example.

\begin{example}
Let $\code\subseteq (\mathbb{Z}/9\mathbb{Z})^3$ be the code generated by
\begin{equation*}
    G=\begin{pmatrix}
1 & 2 & 1\\
0 & 3 & 0
\end{pmatrix},
\end{equation*}
and let $\tuple{a}=(1,1,1)$. We compute the $(\tuple{a},1)$-th $\Ring$-binomial moment of\/ $\code$. By Theorem~\ref{prop:antiunique}, every anticode in
$\mathcal A_{\tuple a}(\Ring^3)$ is permutation equivalent to the code generated by
\begin{equation*}
\begin{pmatrix}
1 & 0 & 0\\
0 & 3 & 0
\end{pmatrix}.
\end{equation*}
In particular, the family $\mathcal A_{\tuple a}(\Ring^3)$ consists of the following
six anticodes:
\begin{equation*}
\begin{array}{ccc}
\mathcal A_1=\left\langle
\begin{pmatrix} 1 & 0 & 0\\ 0 & 3 & 0 \end{pmatrix}
\right\rangle,
&
\mathcal A_2=\left\langle
\begin{pmatrix} 1 & 0 & 0\\ 0 & 0 & 3 \end{pmatrix}
\right\rangle,
&
\mathcal A_3=\left\langle
\begin{pmatrix} 0 & 1 & 0\\ 3 & 0 & 0 \end{pmatrix}
\right\rangle,\\[8pt]
\mathcal A_4=\left\langle
\begin{pmatrix} 0 & 1 & 0\\ 0 & 0 & 3 \end{pmatrix}
\right\rangle,
&
\mathcal A_5=\left\langle
\begin{pmatrix} 0 & 0 & 1\\ 3 & 0 & 0 \end{pmatrix}
\right\rangle,
&
\mathcal A_6=\left\langle
\begin{pmatrix} 0 & 0 & 1\\ 0 & 3 & 0 \end{pmatrix}
\right\rangle.
\end{array}
\end{equation*}

One can easily check that $\code\cap\mathcal A_1=\code\cap\mathcal A_6=\langle(0,3,0)\rangle$, $\code\cap\mathcal A_2=\code\cap\mathcal A_5=\langle(3,0,3)\rangle$, and $\code\cap\mathcal A_3=\code\cap\mathcal A_4=\langle(0,0,0)\rangle$.
Consequently, we get
\begin{equation*}
B^{(1)}_{\mathcal A_1}(\code)=B^{(1)}_{\mathcal A_2}(\code)
=B^{(1)}_{\mathcal A_5}(\code)=B^{(1)}_{\mathcal A_6}(\code)=1,
\qquad
B^{(1)}_{\mathcal A_3}(\code)=B^{(1)}_{\mathcal A_4}(\code)=0,   
\end{equation*}
and therefore $B^{(1)}_{\tuple a}(\code)=4$. On the other hand, for $i\in\{1,2,5,6\}$, we have
    
\begin{equation*}
\sum_{\substack{\tuple b\in\Delta_3(3)\\
\tuple b\preceq\subtype(\code\cap\mathcal A_i)\\
b_s=2}}\gb{\subtype(\code\cap\mathcal A_i)}{\tuple b}_p=
\sum_{\substack{\tuple b\in\Delta_3(3)\\
\tuple b\preceq(0,1,2)\\b_s=2}}\gb{(0,1,2)}{\tuple b}_p=
\gb{(0,1,2)}{(0,1,2)}_p
=1,
\end{equation*}
while for $j\in\{3,4\}$, we get
    
\begin{equation*}
\sum_{\substack{\tuple b\in\Delta_3(3)\\
\tuple b\preceq\subtype(\code\cap\mathcal A_j)\\
b_s=2}}\gb{\subtype(\code\cap\mathcal A_j)}{\tuple b}_p=
\sum_{\substack{\tuple b\in\Delta_3(3)\\\tuple b\preceq(0,0,3)\\b_s=2}}\gb{(0,0,3)}{\tuple b}_p
=0.
\end{equation*}
This is in line with Remark~\ref{rem:binom}.
\end{example}

We now show that the $\Ring$-weight distribution and the set of $\Ring$-binomial moments of~$\code$ encode the same information. Before proving the result, we establish the following preliminary lemma. Its proof follows directly from Propositions~\ref{prop:antiunique} and~\ref{prop:order}, and include one for completeness.

\begin{lemma}\label{lem:AA'}
Let $\tuple{a}, \tuple{b} \in \Delta_{s+1}(n)$ with $\tuple{b} \preceq \tuple{a}$, and write $\tuple{a} = (a_0, \ldots, a_s)$ and $\tuple{b} = (b_0, \ldots, b_s)$. Let $\widehat{\tuple{a}}$ and $\widehat{\tuple{b}}$ denote the associated sequences of $\tuple{a}$ and $\tuple{b}$, respectively. Then   
     \begin{equation*}
         |\set{(\mathcal{A},\mathcal{A}')\in A_\tuple{a}\times A_\tuple{b} \mid \mathcal{A}'\leq \mathcal{A}}|=\prod_{j=0}^{s}\binom{n-\widehat{a}_j}{a_j}\binom{\widehat{a}_j-\widehat{b}_{j-1}}{b_j}.
     \end{equation*}
\end{lemma}
\begin{proof}
The main idea of the proof is to count the number of possible mappings from a support subtype $\tuple{a}$ of a code $\mathcal{A} \in A_{\tuple{a}}$ to a support subtype $\tuple{b}$ of a code $\mathcal{A}' \in A_{\tuple{b}}$. The counting proceeds iteratively by comparing, for each $j \in \{0, \ldots, s\}$, the first $j$ entries of $\tuple{a}$ with those of $\tuple{b}$ and counting the number of admissible mappings of the rows of a generator matrix $G_{\mathcal{A}'}$ of $\mathcal{A}'$ from those of a generator matrix $G_\mathcal{A}$ of $\mathcal{A}$.
We begin by considering $a_0$ and $b_0$. If $a_0 = b_0$, there is exactly one way to map the first $a_0$ rows of a generator matrix of $\mathcal{A}$ to those of a generator matrix of $\mathcal{A}'$, namely via the identity map. If instead $a_0 > b_0$, only $b_0$ of the first $a_0$ rows of a generator matrix of $\mathcal{A}$ are preserved. In both cases, there are $\binom{a_0}{b_0}$ possible choices.
Now let $j \in \{1, \ldots, s\}$, and let $\widehat{\tuple{a}}$, respectively $\widehat{\tuple{b}}$, denote the associated sequence of $\tuple{a}$, respectively $\tuple{b}$. Since $\tuple{b} \preceq \tuple{a}$, we have $\widehat{a}_j \geq \widehat{b}_j$, which implies that $b_j \leq \widehat{a}_j - \widehat{b}_{j-1}$. Hence, there are $\binom{\widehat{a}j - \widehat{b}{j-1}}{b_j}$ possible mappings at step $j$. Continuing this argument for each $j \in \{0, \ldots, s\}$ yields a total of
    \begin{align}
        \prod_{j = 0}^{s} \binom{\widehat{a}_j - \widehat{b}_{j-1}}{b_j}
    \end{align}
distinct codes $\mathcal{A}' \in A_{\tuple{b}}$ corresponding to a fixed code $\mathcal{A} \in A_{\tuple{a}}$. Finally, by a similar argument, the number of possible codes $\mathcal{A} \in A_{\tuple{a}}$ is given by
\begin{align}
\prod_{j=0}^{s} \binom{n - \widehat{a}_{j-1}}{a_j}.
\end{align}
Multiplying these two quantities yields the desired result.
\end{proof}

The following result is the analogue of~\cite[Theorem~6.7]{cotardo2022zeta} (cf.~\cite[Theorem~6.4]{byrne2023tensor}). The proof is similar, but we include it for completeness.

\begin{theorem}\label{thm:bindistr}
    Let $\tuple{a}\in\Delta_{s+1}(n)$, $\widehat{\tuple{a}}$ its associated sequence, and $j\in\{1,\ldots,K\}$. We recall that $\mu_{s+1}(n;\cdot,\cdot)$ denotes the M\"obius function of the lattice $\Delta_{s+1}(n)$. The following hold.
    \begin{enumerate}
        \item $\displaystyle B_{\tuple{a}}^{(j)}(\code)=\sum_{\substack{\tuple{b}\in\Delta_{s+1}(n)\\ \tuple{b}\preceq\tuple{a}}}W_{\tuple{b}}^{(j)}(\code)\prod_{i=0}^{s}\binom{n-\widehat{a}_{i-1}}{a_i}\binom{\widehat{a}_{i}-\widehat{b}_{i-1}}{b_i}$.
        \item $\displaystyle W_{\tuple{a}}^{(j)}(\code)=\sum_{\substack{\tuple{b}\in\Delta_{s+1}(n)\\\tuple{b}\preceq\tuple{a}}}\mu_{s+1}(n;\tuple{b},\tuple{a})B_{\tuple{b}}^{(j)}(\code)\prod_{i=0}^{s}\binom{n-\widehat{a}_{i-1}}{a_i}\binom{\widehat{a}_{i}-\widehat{b}_{i-1}}{b_i}$.
    \end{enumerate}
\end{theorem}
\begin{proof}
    For ease of notation, throughout the proof we omit the explicit dependence on $\code$. For example, we write $B_{\tuple{a}}^{(j)}$ instead of $B_{\tuple{a}}^{(j)}(\code)$. Observe that $B_{(0,\ldots,0)}^{(j)} = W_{(0,\ldots,0)}^{(j)} = 0$ since $j \geq 1$. Hence, we can assume that $(0,\ldots,0) \preceq  \tuple{b}\preceq \tuple{a}$ in the remainder.  Observe that for any $\mathcal{A} \in \mathcal{A}(\Ring^n)$, we have
    \begin{equation}\label{eq:sumWA}
        \sum_{\substack{\mathcal{A}'\in\mathcal{A}(\Ring^n)\\\mathcal{A}'\leq \mathcal{A}}}W_{\mathcal{A}'}^{(j)}=|\{\mathcal{D}\leq \code\cap \mathcal{A} \mid \rk(\mathcal{D})=j\}|=B_{\mathcal{A}}^{(j)}
    \end{equation}
    and by M\"obius inversion we get
    \begin{equation}\label{eq:sumBA}
        W_\mathcal{A}^{(j)}=\sum_{\substack{\mathcal{A}'\in\mathcal{A}(\Ring^n)\\\mathcal{A}'\leq \mathcal{A}}}\mu_s(n;\subtype(\mathcal{A}'),\subtype(\mathcal{A}))B_{\mathcal{A}'}^{(j)},
    \end{equation}
    where $\subtype(\mathcal{A})=(c_0,\ldots,c_s)$ and $(c_0,\ldots,c_{s-1})$ is the support subtype of $\mathcal{A}$ and $n-c_s=\rk(\mathcal{A})$, that is $\subtype(\mathcal{A})=\min\{\tuple{c}\in\Delta_{s+1}(n):\mathcal{A}\in\mathcal{A}_\tuple{c}\}$.
    Therefore, by Equation~\ref{eq:sumWA} and Lemma~\ref{lem:AA'}, we have
    \begin{align*}
        B_{\tuple{a}}^{(j)}&=\sum_{\mathcal{A}\in\mathcal{A}_\tuple{a}(\Ring^n)}B_\mathcal{A}^{(j)}=\sum_{\mathcal{A}\in\mathcal{A}_\tuple{a}(\Ring^n)}\sum_{\substack{\mathcal{A}'\in\mathcal{A}_\tuple{b}\\\mathcal{A}'\leq \mathcal{A}}}W_{\mathcal{A}'}^{(j)}\\
        &=\sum_{\substack{\tuple{b}\in\Delta_{s+1}(n)\\\tuple{b}\preceq\tuple{a}}}\sum_{\mathcal{A}'\in\mathcal{A}_\tuple{b}}W_{\mathcal{A}'}^{(j)}|\{\mathcal{A}\in\mathcal{A}_\tuple{a}(\Ring^n)\mid \mathcal{A}'\leq \mathcal{A}\}|\\
        &=\sum_{\substack{\tuple{b}\in\Delta_{s+1}(n)\\\tuple{b}\preceq\tuple{a}}}\prod_{i=0}^{s}\binom{n-\widehat{a}_{i-1}}{a_i}\binom{\widehat{a}_{i}-\widehat{b}_{i-1}}{b_i}\sum_{\mathcal{A}'\in\mathcal{A}_\tuple{b}}W_{\mathcal{A}'}^{(j)}\\
        &=\sum_{\substack{\tuple{b}\in\Delta_{s+1}(n)\\\tuple{b}\preceq\tuple{a}}}W_\tuple{b}^{(j)}\prod_{i=0}^{s}\binom{n-\widehat{a}_{i-1}}{a_i}\binom{\widehat{a}_{i}-\widehat{b}_{i-1}}{b_i}.
    \end{align*}
    On the other hand, by~\ref{eq:sumBA} and Lemma~\ref{lem:AA'}, we have
    \begin{align*}
        W_\tuple{a}^{(j)}&=\sum_{\mathcal{A}\in\mathcal{A}_\tuple{a}(\Ring^n)}W_\mathcal{A}^{(j)}=\sum_{\mathcal{A}\in\mathcal{A}_\tuple{a}(\Ring^n)}\sum_{\substack{\mathcal{A}'\in\mathcal{A}(\Ring^n)\\\mathcal{A}'\leq \mathcal{A}}}\mu_s(n;\subtype(\mathcal{A}'),\subtype(\mathcal{A}))B_{\mathcal{A}'}^{(j)}\\
        &=\sum_{\substack{\tuple{b}\in\Delta_{s+1}(n)\\\tuple{b}\preceq\tuple{a}}}\mu_s(n;\tuple{b},\tuple{a})\sum_{\mathcal{A}'\in\mathcal{A}_\tuple{b}}B_{\mathcal{A}'}^{(j)}|\{\mathcal{A}\in\mathcal{A}_\tuple{a}(\Ring^n)\mid \mathcal{A}'\leq \mathcal{A}\}|\\
        &=\sum_{\substack{\tuple{b}\in\Delta_{s+1}(n)\\\tuple{b}\preceq\tuple{a}}}\mu_s(n;\tuple{b},\tuple{a})\prod_{i=0}^{s}\binom{n-\widehat{a}_{i-1}}{a_i}\binom{\widehat{a}_{i}-\widehat{b}_{i-1}}{b_i}\sum_{\mathcal{A}'\in\mathcal{A}_\tuple{b}}B_{\mathcal{A}'}^{(j)}\\
        &=\sum_{\substack{\tuple{b}\in\Delta_{s+1}(n)\\\tuple{b}\preceq\tuple{a}}}\mu_s(n;\tuple{b},\tuple{a})B_{\tuple{b}}^{(j)}(\code)\prod_{i=0}^{s}\binom{n-\widehat{a}_{i-1}}{a_i}\binom{\widehat{a}_{i}-\widehat{b}_{i-1}}{b_i}.
    \end{align*}
    This concludes the proof.
\end{proof}

The next result is the analogue of~\cite[Lemma~28]{ravagnani2016rank} and~\cite[Lemma~6.5]{byrne2023tensor} for codes over~$\Ring$.

\begin{lemma}\label{lem:codedual}
     For any $\tuple{a}\in\Delta_{s+1}(n)$ and $\mathcal{A}\in\mathcal{A}_\tuple{a}(\Ring^n)$, we have  
     \begin{equation*}
         \text{rk}( \code \cap \mathcal{A})= K-a_s+\text{freerk}(\code^\perp \cap \mathcal{A}^\perp)
     \end{equation*}
\end{lemma}
\begin{proof}
    Using the well-known fact that $\rk(\mathcal{C})=n-\text{freerk}(\mathcal{C}^\perp)$, we have 
    \begin{align*}
        \rk(\code\cap \mathcal{A})&=n-\text{freerk}((\mathcal{C}\cap \mathcal{A})^\perp)=n-\text{freerk}(\mathcal{C}^\perp+ \mathcal{A}^\perp)\\
        &=n-\text{freerk}(\mathcal{C}^\perp)-\text{freerk}(\mathcal{A}^\perp)+\text{freerk}(\mathcal{C}^\perp\cap \mathcal{A}^\perp)\\
        &=n-(n-K)-a_s+\text{freerk}(\mathcal{C}^\perp\cap \mathcal{A}^\perp)
    \end{align*}
    which concludes the proof.
\end{proof}

\begin{remark}
   For the Hamming metric, one can easily observe that, by the definition of an $\Ring$-anticode, if
$\rk(\mathcal{A}) = n - a_s < \HD(\code)$ for some $\tuple{a} \in \Delta_{s+1}(n)$, then
$B_{\tuple{a}}^{(j)}(\code) = W_{\tuple{a}}^{(j)}(\code) = 0$ for all
$j \in \{1, \ldots, K\}$.
On the other hand, if $\rk(\mathcal{A}^\perp) = n-a_0 < \HD(\code^\perp)$, then
$\text{freerk}(\code^\perp \cap \mathcal{A}^\perp) = 0$, and hence $\rk(\code \cap \mathcal{A}) = K - a_s$.
\end{remark} 

We introduce a new invariant for codes over $\Ring$. 
We recall that by the order-extension principle, we can extend the partial order $\preceq$ on $\Delta_{s+1}(n)$  to a total order. In the following, we assume a total order $\leq$ on  $\Delta_{s+1}(n)$ that extends $\preceq$.

\begin{definition}
    For any positive integer $r$, the $r$-\textit{th}  $\Ring$-\textit{weight} of $\code$ is
    \begin{equation*}
        \dist_r(\mathcal{C}) = \min \{\tuple{a}\in \Delta_{s+1}(n)\mid \textup{there exists } \mathcal{A}\; \in \mathcal{A}_\tuple{a}(\Ring^n) \textup{ with } \text{rk}(\mathcal{A} \cap \mathcal{C}) \geq r\}.
    \end{equation*}
\end{definition}

\begin{remark}\label{rem:dr}
    It is immediate to check that they form an increasing chain, that is $\dist_r(\mathcal{C})\leq \dist_{r+1}(\mathcal{C})$, since any anticode $\mathcal{A}\in\mathcal{A}_\tuple{a}(\Ring^n)$ with $\rk(\mathcal{A}\cap\mathcal{C})\geq r+1$ satisfies $\rk(\mathcal{A}\cap\mathcal{C})\geq r$.
\end{remark}

We conclude this section by establishing a relation between the $\Ring$-weights and the generalized Hamming weights of a code. We recall that for a positive integer $r$, the $r$-th generalized Hamming weight is
\begin{equation*}
    \HD_r(\mathcal{C})=\min\{\suppH(\mathcal{D}) \mid \mathcal{D}\leq \mathcal{C}, \rk(\mathcal{D})=r\},
\end{equation*}
where $\suppH(\mathcal{D})$ denotes the \textit{Hamming support} of $\mathcal{D}$, that is
\begin{equation*}
    \suppH(\mathcal{D})=\bigcup_{c\in\mathcal{D}}\suppH(c)\quad \text{ with }\quad \suppH(c)=\{i\mid c_i\neq 0\}.
\end{equation*}

The following result is the analogue of~\cite[Theorem~10]{ravagnani2016generalized}. We omit the proof, as it follows the same lines as the proof of \cite[Theorem~10]{ravagnani2016generalized}, together with the fact that the rank and free rank of a code coincide for free codes.

\begin{proposition}
   Let $p\neq 2$ and let $r$ be a positive integer. We have
   \begin{equation*}
       \HD_r(\mathcal{C})=\min\{\rk(\mathcal{A})\mid \mathcal{A}\in  \mathcal{A}_\tuple{a}(\Ring^n) \textup{ is free with }\text{rk}(\mathcal{A} \cap \mathcal{C}) \geq r\}.
   \end{equation*}
\end{proposition}

We recall that $\mathcal{A}$ is a free anticode if and only if its  support subtype is $(\rk(\mathcal{A}),0,\ldots,n-\rk(\mathcal{A}))$. As a consequence, we have that the generalized Hamming weights are in one-to-one correspondence with 
\begin{equation*}
     \dist_r^\text{free}(\mathcal{C}) = \min \{\tuple{a}\in \Delta_{s+1}(n)\mid \textup{there exists } \mathcal{A}\; \in \mathcal{A}_\tuple{a}(\Ring^n) \textup{ is free with } \text{rk}(\mathcal{A} \cap \mathcal{C}) \geq r\}
\end{equation*}
for any positive integer $r$. We clearly have $\dist_r(\mathcal{C})\leq \dist_r^\text{free}(\mathcal{C})$ for any $r\in\{1,\ldots,K\}$.  As a consequence of this fact, ~\cite[Theorems~1 and~3]{wei2002generalized}, and~\cite[Theorem~10]{ravagnani2016generalized}, for any $p\neq 2$, we have that the $\dist_r^\text{free}(\mathcal{C})$ form a strictly increasing chain, that is $\dist_r^\text{free}(\mathcal{C})< d_{r+1}^\text{free}(\mathcal{C})$ for any $r\in\{1,\ldots,K-1\}$ and that the following sets  determine each other uniquely:
\begin{equation*}
    \left\{\dist_r^\text{free}(\mathcal{C})\mid r\in\{1,\ldots, K\}\right\}\quad\textup{ and }\quad \left\{\dist_r^\text{free}(\mathcal{C}^\perp)\mid r\in\{1,\ldots, n-K\}\right\}.
\end{equation*}

\begin{proposition}
    Let $\mathcal{C},\mathcal{D}\subseteq\Ring^n$ be $\Ring$-linear codes of $\Ring$-rank $K$, and let $r\in\{1,\ldots,K\}$. 
    If $\dist_r(\mathcal{C})=\dist_r(\mathcal{D})$, then $\dist_r^{\textup{free}}(\mathcal{C})=\dist_r^{\textup{free}}(\mathcal{D})$.
\end{proposition}

\begin{proof}
    Define the set
\begin{equation*} S_r(\mathcal{C})=\{\tuple{a}\in \Delta_{s+1}(n)\mid \textup{there exists } \mathcal{A}\in \mathcal{A}_\tuple{a}(\Ring^n) \textup{ with } \text{rk}(\mathcal{A} \cap \mathcal{C}) \geq r\} \end{equation*}
    By definition, the $r$-th  $\Ring$-weight of\/ $\mathcal{C}$ satisfies $\dist_r(\mathcal{C})=\min S_r(\mathcal{C})$. Observe that if $\tuple a=(a_0,\ldots,a_s)\in S_r(\mathcal{C})$, then the
    tuple
    \begin{equation*} \left(\textstyle\sum_{i=0}^{s-1} a_i,0,\ldots,0,n-\sum_{i=0}^{s-1} a_i\right)
    \end{equation*}
    also belongs to $S_r(\mathcal{C})$. Indeed, every ring-linear anticode is contained in a free anticode of the same
    rank, and such containment preserves the rank of the intersection with
    $\mathcal{C}$. One can easily verify that, by definition, we have that
    \[
        \dist_r^{\textup{free}}(\mathcal{C})
        =
        \min\left\{S_r(\mathcal{C})\cap\{\mathcal{A}\in \mathcal{A}_\tuple{a}(\Ring^n)\mid \mathcal{A} \textup{ is a free code}\}\right\}.
    \]
    The same reasoning applies to $\mathcal{D}$. Since
    $\dist_r(\mathcal{C})=\dist_r(\mathcal{D})$, the minimum of $S_r(\mathcal{C})$ and
    $S_r(\mathcal{D})$ coincide, and hence their restrictions to free anticodes have
    the same minimum. This concludes the proof.
\end{proof}

\section*{Acknowledgment}
J.~Bariffi is funded by the European Union (DiDAX, 101115134). Views and opinions expressed are however those of the authors only and do not necessarily reflect those of the European Union or the European Research Council Executive Agency. Neither the European Union nor the granting authority can be held responsible for them. V. Weger  wants to thank the support of the Technical University of Munich – Institute for Advanced Study, funded by the German Excellence Initiative. 

\bibliographystyle{plain}
\bibliography{references.bib}

\end{document}